\documentclass[11pt]{article}
\usepackage[utf8]{inputenc}
\usepackage[american]{babel}
\usepackage[textwidth=6.5in,textheight=9in]{geometry}
\usepackage{setspace}
\usepackage{amsmath,amssymb,amsthm}
\usepackage{natbib}
\usepackage[colorlinks=true,citecolor=blue]{hyperref}

\usepackage{tikz}
\usetikzlibrary{positioning}
\usetikzlibrary{arrows}
\tikzset{
    > = stealth,
    every node/.append style = {
        draw = none,
        text = black
    },
    every path/.append style = {
        arrows = ->
    },
    hidden/.style = {
        draw = black,
        shape = circle,
        inner sep = 2pt
    }
}

\newtheorem{assumption}{Assumption}
\newtheorem{theorem}{Theorem}
\newtheorem{condition}{Condition}
\newtheorem{lemma}{Lemma}

\theoremstyle{remark}
\newtheorem{remark}{Remark}

\begin{document}

\title{A Semiparametric Instrumented Difference-in-Differences Approach to Policy Learning}

\author{Pan Zhao$^1$\thanks{Email: \href{mailto:pan.zhao@inria.fr}{pan.zhao@inria.fr}}, Yifan Cui$^2$\thanks{Email: \href{mailto:cuiyf@zju.edu.cn}{cuiyf@zju.edu.cn}} \\
    $^1$ PreMeDICaL, Inria \\
    Desbrest Institute of Epidemiology and Public Health, University of Montpellier, France \\
    $^2$ Center for Data Science, Zhejiang University, China}
\maketitle
 
\begin{abstract}
Recently, there has been a surge in methodological development for the difference-in-differences (DiD) approach to evaluate causal effects. Standard methods in the literature rely on the parallel trends assumption to identify the average treatment effect on the treated. However, the parallel trends assumption may be violated in the presence of unmeasured confounding, and the average treatment effect on the treated may not be useful in learning a treatment assignment policy for the entire population. In this article, we propose a general instrumented DiD approach for learning the optimal treatment policy. Specifically, we establish identification results using a binary instrumental variable (IV) when the parallel trends assumption fails to hold. Additionally, we construct a Wald estimator, novel inverse probability weighting (IPW) estimators, and a class of semiparametric efficient and multiply robust estimators, with theoretical guarantees on consistency and asymptotic normality, even when relying on flexible machine learning algorithms for nuisance parameters estimation. Furthermore, we extend the instrumented DiD to the panel data setting. We evaluate our methods in extensive simulations and a real data application.
\end{abstract}

\vfill
\newpage

\section{Introduction}

Data-driven individualized decision making has received increasing interests in many fields, such as precision medicine \citep{luedtke2016statistical,tsiatis2019dynamic}, econometrics and  quantitative social sciences \citep{imai2004causal,athey2021policy}, computer science and operations research \citep{shi2022off,kallus2022doubly}. The common goal is to learn optimal treatment assignment policies (also known as regimes, rules or plans) which map individual characteristics to treatment assignments so as to optimize some functional of the counterfactual outcome distributions, leveraging observational data where causal effects can be identified under various strategies and assumptions.

Popular existing methods in the statistical and machine learning literature include model-based approaches such as Q-learning \citep{watkins1992q,murphy2003optimal,linn2017interactive}, A-learning \citep{robins2000marginal,shi2018high}, and direct model-free policy search approaches \citep{zhang2012estimating,zhao2012estimating}. Recent advances of policy learning have also considered a variety of data structures, optimization objectives, criteria or constraints, such as survival and longitudinal data \citep{goldberg2012q,ertefaie2018constructing,zhao2023efficient}, networks \citep{viviano2019policy,sherman2020general}, distributional robustness \citep{mo2021learning,sahoo2022learning}, budget, fairness, or interpretability constraints \citep{luedtke2016optimal,fang2022fairness}, among others \citep{luedtke2020performance,hadad2021confidence,nie2021learning,hu2022fast,jin2023sensitivity}.

With few exceptions, most methods in prior work rely on the pivotal assumption that there is no unmeasured confounding. This is a key threat to credible causal inference in observational studies, and may lead to suboptimal policies, because this assumption is impossible to verify or test in practice. An ad hoc work-around commonly adopted by practitioners is to collect and appropriately adjust for a large number of covariates, which still lacks theoretical guarantee and seems likely to be error-prone. To address this limitation, there has been recent progress made in several directions. \cite{kallus2018confounding} propose to minimize the worst-case regret of a policy under a marginal sensitivity model for the unmeasured confounding. \cite{zhang2021selecting} utilize a randomization test to rank by a partial order and select treatment rules within a given finite collection. While partial identification results provide certain improvement, the performance of such a learned policy may still be suboptimal. \cite{qi2023proximal} build on the semiparametric proximal causal inference framework introduced by \cite{cui2023semiparametric} to establish point identification results on different policy classes and accordingly propose several classification-based approaches; but this framework requires the analyst to correctly classify the measured covariates into three types of proxies, and it may be difficult to estimate the confounding bridge functions.

Instrumental variable methods are widely used to handle unmeasured confounding in observational studies or randomized trials with non-compliance. The core requirements for a pretreatment variable to be a valid IV are: (i) it is associated with the treatment; (ii) it is independent of all unmeasured confounders; (iii) it does not have a direct causal effect on the outcome other than through the treatment. Along with the seminal work of \cite{imbens1994identification,angrist1996identification}, extensive development has been made in using the IV to estimate the local average treatment effect \citep{tan2006regression,ogburn2015doubly}, defined as the average treatment effect for the complier subgroup who would always comply with their treatment assignments. Since the complier subgroup is unknown and may have systematically different characteristics from the population, the population (conditional) average treatment effect is arguably the causal parameter of primary interest in most studies \citep{hernan2006instruments,aronow2013beyond}, especially for policy learning. More recently, \cite{pu2021estimating} consider a partial identification approach to optimal treatment rule estimation; and \cite{wang2018bounded} formally establish point identification of the population average treatment effect under alternative no-interaction assumptions, upon which \cite{cui2021semiparametric} propose various IV methods for estimating optimal treatment regimes. It is notable that all of these IV methods in the literature only consider the setting with a single time point, with the only exception of \cite{pmlr-v202-xu23x}, where the authors propose an IV approach to off-policy evaluation in confounded Markov decision processes with infinite horizons.

There has always been interest in exploiting the longitudinal structure common in datasets such as electronic health records and medical claims in epidemiology and biomedicine \citep{robins2000marginal}, as well as cross-sectional or panel data in program evaluations, economic censuses, and surveys \citep{athey2017state}. DiD methods have been an important tool widely used by empirical researchers \citep{card1994minimum}. The key identification assumption of DiD is that the trend in outcome of the control group over time is informative
about what the trend would have been for the treatment group in the absence of the
treatment. Specifically, under the standard (conditional) parallel trends assumption, which states that the (conditional) expected trends in the potential outcomes of the two groups in the absence of the treatment are identical, the average treatment effect on the treated can be identified \citep{abadie2005semiparametric,sant2020doubly}; we refer interested readers to \cite{lechner2011estimation} and \cite{roth2023s} for detailed reviews. However, concerns often arise that the parallel trends assumption may be violated due to unmeasured confounding. \cite{athey2006identification} develop a new changes-in-changes model that relates outcomes to an individual's group, time, and unobservable characteristics; and various recent extensions for DiD include partial identification \citep{ye2020negative}, sensitivity analysis \citep{keele2019patterns} and negative control \citep{sofer2016negative}, among others \citep{dukes2022semiparametric,park2023universal}. Moreover, DiD methods focus on the identification and estimation of the average treatment effect on the treated, which limits its application in policy learning since the treated cannot represent the population. To the best of our knowledge, this is the first work to systematically study policy learning under the DiD setting.

In this article, we combine the two natural experiments and propose an instrumented DiD approach to policy learning when the parallel trends assumption fails to hold in the presence of unmeasured confounding. Specifically, we adapt and extend the recent progress in \cite{ye2022instrumented} and \cite{vo2022structural}, relaxing some key assumptions of the conventional IV and DiD methods. We allow for the violation of the parallel trends assumption by leveraging an IV which has no direct effect on the the trend in outcome, and does not modify the average treatment effect. Notably, this exogenous variable is not necessarily a valid instrument for the conventional treatment-outcome association, since we allow it to have a direct effect on the outcome not just through the treatment at each time point.

The contributions of this article are summarized as follows. First, we propose the direct policy search approach to learn optimal treatment assignment policies, based on the conditional average treatment effect estimators using instrumented DiD. This approach essentially allows us to learn the optimal policy that maximizes the estimated value within a restricted policy class. Second, we establish novel identification results of optimal policies for the instrumented DiD design subject to unmeasured confounding. The new results give rise to new inverse probability weighting estimators of optimal policies without necessarily identifying the value function for a given policy. Another interesting progress is also made towards identifying optimal policies without necessarily using the subjects' realized treatment values. In summary, we construct a Wald estimator and novel inverse probability weighting estimators. A class of semiparametric efficient and multiply robust estimators is also proposed, which is consistent provided that a subset of several posited models indexing the observed data distribution is correctly specified. Third, we prove theoretical guarantees for the proposed multiply robust policy learning approaches. Specifically, we consider both parametric models and flexible data-adaptive machine learning algorithms with the cross-fitting procedure to estimate the nuisance parameters, to draw valid inferences under mild regularity conditions and certain rate of convergence conditions. In particular, we consider a restricted policy class indexed by an Euclidean parameter $\eta$ and establish the $n^{-1/3}$ convergence rate of $\hat{\eta}$, even though its resultant limiting distribution is not standard. Fourth, we extend our proposed methods to the panel data setup. We establish identification of the conditional average treatment effect under alternative assumptions and provide the direct policy search approaches for panel data. The theoretical results for panel data can be similarly derived.

The rest of this article is organized as follows. In Section~\ref{sec:sf}, we introduce the statistical framework of instrumental variable, DiD and policy learning. Section~\ref{sec:idid} develops our main methodology of learning the optimal policy using the instrumented DiD. Semiparametric efficiency results and multiply robust estimators are presented in Section~\ref{sec:semr}. Section~\ref{sec:asym} establishes the asymptotic properties of the proposed estimators. Extensive simulations are reported in Section~\ref{sec:simu} to demonstrate the proposed methods, followed by a real data application in Section~\ref{sec:da}. Next, we consider the extension of our methods to panel data in Section~\ref{sec:panel}. The article concludes in Section~\ref{sec:disc} with a discussion of some remarks and future work. All proofs and additional results are provided in the Supplementary Material. 

\section{Statistical framework}
\label{sec:sf}

We first introduce some notation. Let $X$ denote the $p$-dimensional vector of covariates that belongs to a covariate space $\mathcal{X} \subset \mathbb{R}^p$, $A \in \mathcal{A} = \{0, 1\}$ denote the binary treatment, $Y \in \mathbb{R}$ denote the outcome of interest, and $T \in \mathcal{T} = \{0, 1\}$ denote the time period. Suppose that $U = (U_{0}, U_{1})$ is an unmeasured confounder of the effect of $A$ on $Y$, and $Z \in \{0, 1\}$ is a binary instrumental variable; the observed data are $O = (X, A, Y, T, Z)$. We assume that the random samples $(O_{1}, \ldots, O_{n})$ collected at the two time periods are independent and identically distributed (i.i.d.) observations of $O \sim P_{0}$, and there is no overlap between individuals in these two time periods. This setup is commonly known as the repeated cross-sectional data. Extension to panel data setting is studied in Section~\ref{sec:panel}.

We use the potential outcomes framework \citep{neyman1923applications,rubin1974estimating} to define causal effects. Let $A_{t}(z)$ denote the potential exposure at time $t$ if the instrument were set to level $z$, $Y_{t}(a)$ denote the potential outcome at time $t$ if the exposure were set to level $a$ and the instrument would take the same value it actually had, and $Y_{t}(z, a)$ denote the potential outcome at time $t$ had the instrument and exposure been set to $z, a$ respectively.

Without loss of generality, we assume that larger values of $Y$ are more desirable. Our aim is to identify and estimate an policy $d: \mathcal{X} \to \mathcal{A}$, that maximizes the expected potential outcome in a counterfactual world had this policy been implemented on the population. The optimal policy at time $t$ is given by $d_{{\rm opt},t} (x) = I\{\tau_{t} (x) > 0\}$, where $\tau_{t} (x) = E[Y_{t}(1) - Y_{t}(0) \mid X = x]$ is the conditional average treatment effect (CATE) at time $t$.

Let $Y_{t}(d) = d(X) Y_{t}(1) + (1 - d(X)) Y_{t}(0)$ denote the potential outcome under a hypothetical intervention that assigns treatment according to policy $d$. The value function of a policy $d$ at time $t$ is defined as $V_{t}(d) = E[Y_{t}(d)]$. Let $\mathcal{D}$ be the class of candidate policies of primary interest. The optimal policy can be obtained by directly maximizing the value function:
\begin{equation}\label{eq:opt.policy}
d_{{\rm opt},t} = \arg\max_{d \in \mathcal{D}} V_{t}(d) = \arg\max_{d \in \mathcal{D}} E[\tau_{t} (X) d(X)].
\end{equation}

Throughout this article, we assume that the stable treatment effect over time assumption holds, which says that the CATE does not vary over time, and thus ensures that the optimal policy remains the same between the two time periods. The subscript $t$ is omitted when it is clear from the context.

\begin{remark}
Our proposed instrumented DiD methodology can also be readily formulated in the weighted classification perspective. Pioneered by \cite{zhang2012estimating}, this perspective has been widely used in the biostatistics and precision medicine literature, and enjoys certain robustness empirically. Specifically, the above maximization problem~\eqref{eq:opt.policy} can be transformed into the following equivalent weighted classification problem:
\begin{equation}\label{eq:wei.cla}
d_{\rm{opt}} (x) = \arg\max_{d \in \mathcal{D}} E[W I\{A = d(X)\}],
\end{equation}
where $W$ is regarded as a weight that is motivated by standard outcome regression, inverse probability weighting and doubly robust methods. Many robust classification methods and off-the-shelf implementations can be utilized.
\end{remark}

\section{Instrumented difference-in-differences}
\label{sec:idid}

In this section, we introduce a general instrumented DiD framework for policy learning under endogeneity, and provide novel identification results. Let $\pi (t, z, x) = Pr(T = t, Z = z \mid X = x)$, and for any random variable $C \in \{A, Y\}$, we define $\mu_{C}(t, z, x) = E[C \mid T = t, Z = z, X = x]$, $\delta_{C}(x) = \mu_{C}(1, 1, x) - \mu_{C}(0, 1, x) - \mu_{C}(1, 0, x) + \mu_{C}(0, 0, x)$. We make the following identification assumptions.

\begin{assumption}[Consistency]\label{asmp:cons}
$A = A_{T}(Z)$ and $Y = Y_{T}(A)$.
\end{assumption}

\begin{assumption}[Positivity]\label{asmp:posi}
$c_1 < \pi (t, z, x) < 1 - c_1$ for some $0< c_1 < 1/2$.
\end{assumption}

\begin{assumption}[Random sampling]\label{asmp:ran.samp}
$T \perp \{A_{t}(z), Y_{t}(a) : t = 0, 1, z = 0, 1, a = 0, 1\} \,|\, X, Z$.
\end{assumption}

\begin{assumption}[Stable treatment effect over time]\label{asmp:stab.eff}
$E[Y_{0}(1) - Y_{0}(0) \mid X] = E[Y_{1}(1) - Y_{1}(0) \mid X]$.
\end{assumption}

Assumption~\ref{asmp:cons} is also known as the stable unit treatment value assumption, which states that there is no interference between subjects and no multiple versions of the instrument and treatment. Assumption~\ref{asmp:posi} ensures the same support of $X$ for each $(T, Z)$ level. Assumption~\ref{asmp:ran.samp} is commonly assumed for repeated cross-sectional data \citep{abadie2005semiparametric}. Assumption~\ref{asmp:stab.eff} requires that the CATE $\tau (x)$ does not vary over time, and thus ensures that the optimal policy remains the same between the two time periods.

\begin{assumption}[Trend relevance]\label{asmp:trend.rele}
$E[A_{1}(1) - A_{0}(1) \mid Z = 1, X] \neq E[A_{1}(0) - A_{0}(0) \mid Z = 0, X]$.
\end{assumption}

\begin{assumption}[Independence \& exclusion restriction]\label{asmp:ind.excl}
$Z \perp \{A_{t}(1), A_{t}(0), Y_{t}(1) - Y_{t}(0), Y_{1}(0) - Y_{0}(0) : t = 0, 1\} \,|\, X$.
\end{assumption}

\begin{assumption}[No unmeasured common effect modifier]\label{asmp:no.com.modi}
$Cov\{A_{t}(1) - A_{t}(0), Y_{t}(1) - Y_{t}(0) \mid X\} = 0$ for $t = 0, 1$.    
\end{assumption}

Assumption~\ref{asmp:trend.rele} and \ref{asmp:ind.excl} are parallel to the core assumptions in the standard IV literature. Directed acyclic graphs illustrating the causal structure are provided in Section~\ref{sm.sec:dag} of the Supplementary Material. Assumption~\ref{asmp:trend.rele} states that the IV affects the trend in treatment. Assumption~\ref{asmp:ind.excl} requires that the IV is unconfounded, has no direct effect on the trend in outcome, and does not modify the treatment effect. This exogenous variable is not necessarily a valid instrument for the conventional treatment-outcome association, since we allow it to have a direct effect on the outcome not just through the treatment at each time point. Assumption~\ref{asmp:no.com.modi} essentially states that there is no common effect modifier by an unmeasured confounder, of the additive effect of treatment on the outcome, and the additive effect of the IV on treatment. It has been studied in \cite{cui2021semiparametric}, and relax certain no additive interaction assumptions in \cite{wang2018bounded}. We refer interested readers to \cite{ye2022instrumented} for detailed discussion and concrete examples of an IV for DiD. Now we present our first identification result under the above assumptions.

\begin{theorem}\label{thm:wald}
Under Assumptions~\ref{asmp:cons}-\ref{asmp:no.com.modi}, the optimal policy is nonparametrically identified by
\begin{equation}
\arg\max_{d \in \mathcal{D}} E\left[\frac{\delta_{Y}(X)}{\delta_{A}(X)} d(X)\right].
\end{equation}
\end{theorem}

Theorem~\ref{thm:wald} combines the Wald estimator for CATE and the direct policy search approach in Equation~\eqref{eq:opt.policy}. Similarly, the IPW estimator proposed by \cite{ye2022instrumented} can also be used to learn the optimal policy. Semiparametric efficient and multiply robust estimators are presented in Section~\ref{sec:semr}. Next we propose our novel identification results, which also serves as basis for the estimators proposed in Section~\ref{sec:semr}.

\begin{theorem}\label{thm:ipw1}
Under Assumptions~\ref{asmp:cons}-\ref{asmp:no.com.modi}, the optimal policy is nonparametrically identified by
\begin{equation}\label{eq:ipw1}
\arg\max_{d \in \mathcal{D}} E\left[\frac{(2Z - 1)(2T - 1)(2A - 1) Y I\{A = d(X)\}}{\pi(T, Z, X) \delta_{A}(X)}\right].
\end{equation}
\end{theorem}

Theorem~\ref{thm:ipw1} extends prior identification of CATE, and proposes a novel IPW estimator of the optimal policy without necessarily identifying the value function. Semiparametric efficiency results based on \eqref{eq:ipw1} are given in Section~\ref{sm.sec:semi.eff} and \ref{sm.sec:semi.eff.proof} of the Supplementary Material.

\begin{theorem}\label{thm:ipw2}
Under Assumptions~\ref{asmp:cons}-\ref{asmp:no.com.modi}, the optimal policy is nonparametrically identified by
\begin{equation}\label{eq:ipw2}
\arg\max_{d \in \mathcal{D}} \, E\left[\frac{(2T - 1) Y I\{Z = d(X)\}}{\pi(T, Z, X) \delta_{A}(X)}\right].
\end{equation}
\end{theorem}

Theorem~\ref{thm:ipw2} essentially proves that we can identify the optimal policy without necessarily using the subjects' realized treatment values, for instance when $\delta_{A}(X)$ is known a priori, or when  a separate sample with data on $(A, X, T, Z)$ is available to estimate $\delta_{A}(X)$. To conclude this section, we propose the following estimators for optimal policies:
\begin{align*}
\hat{d}_{\rm Wald} &= \arg\max_{d \in \mathcal{D}} \frac{1}{n}\sum_{i=1}^{n} \frac{\hat{\delta}_{Y}(X_i)}{\hat{\delta}_{A}(X_i)} d(X_i), \\
\hat{d}_{\rm IPW1} &= \arg\max_{d \in \mathcal{D}} \frac{1}{n}\sum_{i=1}^{n} \frac{(2Z_i - 1)(2T_i - 1)(2A_i - 1) Y_i I\{A_i = d(X_i)\}}{\hat{\pi}(T_i, Z_i, X_i) \hat{\delta}_{A}(X_i)}, \\
\hat{d}_{\rm IPW2} &= \arg\max_{d \in \mathcal{D}} \frac{1}{n}\sum_{i=1}^{n} \frac{(2T_i - 1) Y_i I\{Z_i = d(X_i)\}}{\hat{\pi}(T_i, Z_i, X_i) \hat{\delta}_{A}(X_i)},
\end{align*}
where $\hat{\delta}_{Y}$, $\hat{\delta}_{A}$ and $\hat{\pi}$ are estimated by parametric models or machine learning algorithms. Our simulation studies in Section~\ref{sec:simu} empirically shows comparable performance of the IPW estimators~\eqref{eq:ipw1} and \eqref{eq:ipw2}.

\begin{remark}
Similarly, classification-based estimators based on Theorem~\ref{thm:ipw1} and \ref{thm:ipw2} can be proposed:
\begin{equation}\label{eq:opt.policy.equi}
\arg\max_{d \in \mathcal{D}} E[\Tilde{W}_{1} I\{A = d(X)\}], \quad \arg\max_{d \in \mathcal{D}} E[\Tilde{W}_{2} I\{Z = d(X)\}], 
\end{equation}
respectively, where the weights are given by
\begin{equation*}
\Tilde{W}_{1} = \frac{(2Z - 1)(2T - 1)(2A - 1) Y}{\pi(T, Z, X) \delta_{A}(X)}, \quad \Tilde{W}_{2} = \frac{(2T - 1) Y}{\pi(T, Z, X) \delta_{A}(X)}.
\end{equation*}
The Fisher consistency, excess risk bound and universal consistency of the estimated policy can also be established \citep{zhao2012estimating}.
\end{remark}

\section{Semiparametric efficiency and multiply robust estimators}
\label{sec:semr}

In this section, we use semiparametric theory and propose multiply robust estimators. The Wald and the IPW approaches require the corresponding models to be correctly specified. Hence, methods that are robust against model misspecification are highly desired, where consistency is guaranteed when a subset of several posited models indexing the observed data distribution is correctly specified.

We consider the (uncentered) efficient influence function:
\begin{equation*}
\Delta (O) = \frac{\delta_{Y}(X)}{\delta_{A}(X)} + \frac{(2Z - 1)(2T - 1)}{\pi(T, Z, X) \delta_{A}(X)}\left\{Y - \mu_{Y}(T, Z, X) - \frac{\delta_{Y}(X)}{\delta_{A}(X)}(A - \mu_{A}(T, Z, X))\right\},
\end{equation*}
which has been proposed in \cite{ye2022instrumented}. Therefore, the optimal policy is identified by $\arg\max_{\mathcal{D}} E\left[\Delta (X) d(X)\right]$. Moreover, in light of the optimization tasks formulated in \eqref{eq:opt.policy.equi}, we propose the following two choices of statistic:
\begin{equation*}
W_{1} = \frac{(2A - 1)\delta_{Y}(X)}{\delta_{A}(X)} + \frac{(2A - 1)(2Z - 1)(2T - 1)}{\pi(T, Z, X) \delta_{A}(X)}\left\{Y - \mu_{Y}(T, Z, X) - \frac{\delta_{Y}(X)}{\delta_{A}(X)}(A - \mu_{A}(T, Z, X))\right\},
\end{equation*}
and
\begin{equation*}
W_2 = \frac{(2Z - 1)\delta_{Y}(X)}{\delta_{A}(X)} + \frac{2T - 1}{\pi(T, Z, X) \delta_{A}(X)}\left\{Y - \mu_{Y}(T, Z, X) - \frac{\delta_{Y}(X)}{\delta_{A}(X)}(A - \mu_{A}(T, Z, X))\right\},
\end{equation*}
which also enjoy the multiply robustness property.

First, we consider positing parametric models. Let $\mu_{A}(t,z,x;\alpha)$, $\mu_{Y}(t,z,x;\beta)$ and $\pi(t, z, x;\theta)$ denote the posited models. $\hat{\alpha}$, $\hat{\beta}$ and $\hat{\theta}$ can be estimated by maximum likelihood estimation. In Theorem~\ref{thm:mr}, we show the multiple robustness in the sense of maximizing the objective function (or minimizing the weighted classification error) in the union model of the following models:\\
$\mathcal{M}_1$: models for $\pi(t, z, x)$ and $\delta_{A}(x)$ are correct; \\
$\mathcal{M}_2$: models for $\pi(t, z, x)$ and $\delta_{Y}(x) / \delta_{A}(x)$ are correct; \\
$\mathcal{M}_3$: models for $\delta_{Y}(x) / \delta_{A}(x)$ and $\mu_{C}(0, 0, x)$, $\mu_{C}(1, 0, x)$, $\mu_{C}(0, 1, x)$ for $C \in \{A, Y\}$ are correct.

\begin{theorem}\label{thm:mr}
Under Assumptions~\ref{asmp:cons}-\ref{asmp:no.com.modi}, the optimal policy is identified by
\begin{equation}
\arg\max_{\mathcal{D}} \, E\left[W_{1} I\{A = d(X)\}\right] = \arg\max_{\mathcal{D}} \, E\left[W_{2} I\{Z = d(X)\}\right] = \arg\max_{\mathcal{D}} \, E\left[\Delta (X) d(X)\right],
\end{equation}
under the union model $\mathcal{M}_1 \cup \mathcal{M}_2 \cup \mathcal{M}_3$.
\end{theorem}

We also consider using modern machine learning methods to estimate these nuisance parameters. In practice, we apply the cross-fitting technique \citep{schick1986asymptotically,zheng2010asymptotic,chernozhukov2018double}, which is easy to implement. The cross-fitting procedure goes as follows. We randomly split data into $K$ folds; the cross-fitted estimator is given by
\begin{equation*}
\hat{M}_{CF} = \frac{1}{K} \sum_{k=1}^{K} P_{n,k} \{\Delta (O; \hat{\mu}_{A,-k}, \hat{\mu}_{Y,-k}, \hat{\pi}_{-k}) d(X)\},
\end{equation*}
where $P_{n,k}$ denote empirical averages only over the $k$-th fold, and $\hat{\mu}_{A,-k}$, $\hat{\mu}_{Y,-k}$ and $\hat{\pi}_{-k}$ denote the nuisance estimators constructed excluding the $k$-th fold. Similar cross-fitted estimators for $E\left[W_{1} I\{A = d(X)\}\right]$ and $E\left[W_{2} I\{Z = d(X)\}\right]$ can also be constructed in the same way.

\section{Asymptotic analysis of policy learning}
\label{sec:asym}

In this section, we study theoretical guarantees for our proposed policy learning approaches. While researchers have suggested applying machine learning algorithms to estimate the optimal policies from large classes which cannot be described by a finite dimensional parameter \citep{luedtke2016statistical,kunzel2019metalearners}, it is also important to consider certain classes of policies for better interpretability and transparency, especially in clinical medicine and policy research \citep{zhang2015using,athey2021policy}. Specifically, here we focus on a class of feasible policies $\mathcal{D} = \left\{I\{\eta^\top X > 0\} : \eta \in \mathbb{H}\right\}$, where $\eta$ indexes different policies and $\mathbb{H}$ is a compact subset of $\mathbb{R}^{p}$. That is, we analyze the following estimator:
\begin{equation*}
\hat{\eta} = \arg\max_{\eta \in \mathbb{H}} \hat{M}(\eta) = \arg\max_{\eta \in \mathbb{H}} \frac{1}{n} \sum_{i=1}^{n} \hat{\Delta} (O_i) d(X_i;\eta), 
\end{equation*}
where $\hat{M}(\eta)$ is estimated by posited parametric models, or the cross-fitted estimator. Let $\eta^\ast = \arg\max_{\eta \in \mathbb{H}} E[\Delta (X) d(X;\eta)]$ denote the Euclidean parameter that indexes the optimal policy. We detail the main large sample property of our proposed estimator, that $\hat{\eta}$ converges to $\eta^\ast$ at $n^{1/3}$ rate, and that $\hat{M}(\hat{\eta})$ is $n^{1/2}$-consistent and asymptotically normal under weak conditions (mostly requiring standard regularity conditions \citep{white1982maximum}, or only that the nuisance parameters are estimated at faster than $n^{1/4}$ rates).

\begin{remark}
In order to obtain certain rates of convergence or regret bounds, it is necessary to require some control over the complexity of the class $\mathcal{D}$; see \citet[Section~2.2]{athey2021policy} for examples of the VC-dimension of classes of linear rules, decision trees and monotone rules. Here we apply the empirical process techniques to establish theoretical guarantees for linear rules, which also hold on any other $\mathcal{D}$ indexed by finite-dimensional parameters. Also note that all identification and semiparametric efficiency results hold for any class of policies, and other optimization methods can be readily utilized.
\end{remark}

We assume the following regularity conditions.
\begin{condition}\label{cond.policy}
{\rm (i)} The supports of $X$ and $Y$ are bounded. {\rm (ii)} The functions $\mu_{Y}(t,z,x)$, $\mu_{A}(t,z,x)$ and $\pi(t, z, x)$ are smooth and bounded for all $(t,z,x)$. {\rm (iii)} The function $M(\eta)$ is twice continuously differentiable in a neighborhood of $\eta^\ast$; {\rm (iv)} For all $\delta > 0$, we have that $Pr(|X^{T} \eta^\ast| \leq \delta) \leq c_2 \delta$, for some constant $c_2 > 0$ such that $c_2 \delta \leq 1$.
\end{condition}

\begin{condition}\label{cond.para}
{\rm (i)} $\sqrt{n} (\hat{\alpha} - \alpha^\ast) = O_{p}(1)$; {\rm (ii)} $\sqrt{n} (\hat{\beta} - \beta^\ast) = O_{p}(1)$; {\rm (iii)} $\sqrt{n} (\hat{\theta} - \theta^\ast) = O_{p}(1)$.
\end{condition}

\begin{theorem}\label{thm:asym.para}
Under Assumptions~\ref{asmp:cons}-\ref{asmp:no.com.modi}, if Conditions~\ref{cond.policy} and \ref{cond.para} hold, we have {\rm (i)} $\|\hat{\eta} - \eta^\ast\|_{2} = O_{p}(n^{-1/3})$; {\rm (ii)} $\sqrt{n} \{M(\hat{\eta}) - M(\eta^\ast)\} = o_p(1)$; {\rm (iii)} $\sqrt{n} \{\hat{M}(\hat{\eta}) - M(\eta^\ast)\} \to \mathcal{N}(0, \sigma_1^2)$, where $\sigma_1^2$ is given in the Supplementary Material.
\end{theorem}

Condition~\ref{cond.policy} (i), (ii) and (iii) are standard regularity conditions to establish uniform convergence. Condition~\ref{cond.policy} (iv), also known as the margin condition, is often assumed in the literature of classification \citep{tsybakov2004optimal}, reinforcement
learning \citep{hu2022fast} and treatment assignment policies \citep{luedtke2020performance}, to guarantee fast convergence rates. Condition~\ref{cond.para} requires $\sqrt{n}$ convergence rates of parameter estimates of the posited models, which holds under mild conditions.

We assume the following conditions for the machine learning algorithms used to construct cross-fitted estimators.
\begin{condition}\label{cond.np}
$\|\hat{\mu}_{A}(t,z,X) - \mu_{A}(t,z,X)\|_{L_2} = o_{p}(n^{-1/4})$, $\|\hat{\mu}_{Y}(t,z,X) - \mu_{Y}(t,z,X)\|_{L_2} = o_{p}(n^{-1/4})$ and $\|\hat{\pi}(t, z, X) - \pi(t, z, X)\|_{L_2} = o_{p}(n^{-1/4})$, for $t,z = 0,1$.
\end{condition}

\begin{theorem}\label{thm:asym.np}
Under Assumptions~\ref{asmp:cons}-\ref{asmp:no.com.modi}, if Conditions~\ref{cond.policy} and \ref{cond.np} hold, we have {\rm (i)} $\|\hat{\eta} - \eta^\ast\|_{2} = O_{p}(n^{-1/3})$; {\rm (ii)} $\sqrt{n} \{M(\hat{\eta}) - M(\eta^\ast)\} = o_p(1)$; {\rm (iii)} $\sqrt{n} \{\hat{M}(\hat{\eta}) - M(\eta^\ast)\} \to \mathcal{N}(0, \sigma_2^2)$, where $\sigma_2^2$ is given in the Supplementary Material.
\end{theorem}

Condition~\ref{cond.np} says the nuisance estimators must be consistent and converge at a fast enough rate (essentially $n^{1/4}$ in $L_2$ norm). This is quite general and can be achieved by many existing algorithms under nonparametric smoothness, sparsity, or other structural constraints.
According to Theorems~\ref{thm:asym.para} and \ref{thm:asym.np} (ii), the regret of our estimated regime vanishes as the sample size increases. Theorems~\ref{thm:asym.para} and \ref{thm:asym.np} (iii) imply that $\hat M(\hat \eta)$ is a regular and asymptotic normal estimator of $M(\eta^\ast)$.

\section{Simulations}
\label{sec:simu}

In this section, we conduct extensive simulations to evaluate the finite-sample performance of the proposed estimators. Specifically, we compare them to the instrumental variable approach proposed by \cite{cui2021semiparametric}, which is in principle valid only for a single time point. Replication code is available at \href{https://github.com/panzhaooo/policy-learning-instrumented-DiD}{GitHub}.

We first describe the complete data generation process as follows. Baseline covariates $X = (X_{1}, X_{2})^\top$ are generated from independent standard normal distributions. The time period indicator $T$ is generated from a Bernoulli distribution with probability $0.5$. The unmeasured confounders $U = (U_{0}, U_{1})^\top$ are generated from independent bridge distributions with parameter $0.5$ \footnote{The bridge density function is $p(u) = 1 / (2\pi\cosh{(u/2)})$. We use the bridge distribution because by \cite{wang2003matching}, the data generation process ensures that upon marginalizing over $U$, the model for $Pr(A_{t} = 1 \mid Z, X)$ remains a logistic regression.}. The instrumental variable $Z$ is generated from a Bernoulli distribution with probability $0.5$. The potential treatments and outcomes at time points $t = 0, 1$ are generated from the models:
\begin{align*}
Pr(A_{0} = 1 \mid Z, U, X) & = \text{expit}(2 - 7 Z + 0.2 U_{0} + 2 X_{1}), \\
Pr(A_{1} = 1 \mid Z, U, X) & = \text{expit}(-1.5 + 5 Z - 0.15 U_{1} + 1.5 X_{2}), \\
(Y_{0} \mid Z, U, X, A_{0}) & \sim \mathcal{N}(\mu_{0}, 1), \quad (Y_{1} \mid Z, U, X, A_{1}) \sim \mathcal{N}(\mu_{1}, 1),
\end{align*}
where $\mu_{0} = 200 + 10 (A_{0} (1.5 X_{1} + 2 X_{2} - 0.5) + 0.5 U_{0} + 2 Z + 1.5 X_{1} + 2 X_{2})$, and $\mu_{1} = 240 + 10 (A_{1} (1.5 X_{1} + 2 X_{2} - 0.5) + 0.5 U_{1} + 2 Z + 2 X_{1} + 1.5 X_{2})$. Therefore, the optimal policy is $d_{\rm{opt}} (x) = I\{3 x_{1} + 4 x_{2} - 1 > 0\}$. Let $A = T A_{1} + (1 - T) A_{0}, Y = T Y_{1} + (1 - T) Y_{0}$; thus the observed cross-sectional data are $(X, A, Y, T, Z)$.

A large test dataset of size $N = 1 \times 10^{6}$ is generated independently to evaluate the performance of different estimators. The percentage of correct decisions (PCD) of an estimated policy $\hat{d}(x)$ is computed by $1 - N^{-1} \sum_{i=1}^{N} |\hat{d}(X_i) - d_{\rm{opt}}(X_i)|$.

We compare $7$ estimators in our study: the two IPW estimators, the Wald estimator, and the two multiply robust estimators, along with the below IV estimators proposed by \cite{cui2021semiparametric}:
\begin{align*}
d_{\rm IV.t0} &= \arg\max_{d \in \mathcal{D}} \frac{1}{n_{t0}}\sum_{i=1}^{n} \frac{Z_{i} A_{i} Y_{i} I\{A_{i} = d(X_{i})\} I\{T_{i} = 0\}}{\hat{\delta}_{t0}(X_{i}) \hat{\pi}_{t0}(Z_{i},X_{i})}, \\
d_{\rm IV.t1} &= \arg\max_{d \in \mathcal{D}} \frac{1}{n_{t1}}\sum_{i=1}^{n} \frac{Z_{i} A_{i} Y_{i} I\{A_{i} = d(X_{i})\} I\{T_{i} = 1\}}{\hat{\delta}_{t1}(X_{i}) \hat{\pi}_{t1}(Z_{i},X_{i})},
\end{align*}
where $n_{t0}$ and $n_{t1}$ are the sample sizes at time point $0, 1$, respectively; $\delta_{t0}(x) = \mu_{A}(0,1,x) - \mu_{A}(0,0,x)$, $\delta_{t1}(x) = \mu_{A}(1,1,x) - \mu_{A}(1,0,x)$, $\pi_{t0}(z,x) = Pr(Z = z \mid X = x, T = 0)$, $\pi_{t1}(z,x) = Pr(Z = z \mid X = x, T = 1)$ are the nuisance parameters, and $\hat{\delta}_{t0}, \hat{\delta}_{t1}, \hat{\pi}_{t0}, \hat{\pi}_{t1}$ can be estimated using parametric models or machine learning algorithms. We utilize the genetic algorithm implemented in
the R package \texttt{rgenoud} \citep{mebane2011genetic} to solve the optimization tasks. 

First, we posit parametric models for the nuisance parameters. The linear/logistic regression models for $\mu_{A}(t,z,x;\alpha)$, $\mu_{Y}(t,z,x;\beta)$ and $\pi(t,z,x;\theta)$ are correctly specified. The sample size is $n = 5000$.

We also consider flexible machine learning algorithms for nuisance parameter estimation. Specifically, we apply the generalized random forests \citep{athey2019generalized} implemented in the R package \texttt{grf} with default tuning parameters. For the cross-fitting procedure, we use $K = 4$ folds. The sample size is $n = 10^4$.

\begin{figure}[ht]
    \centering
    \includegraphics[scale=0.8]{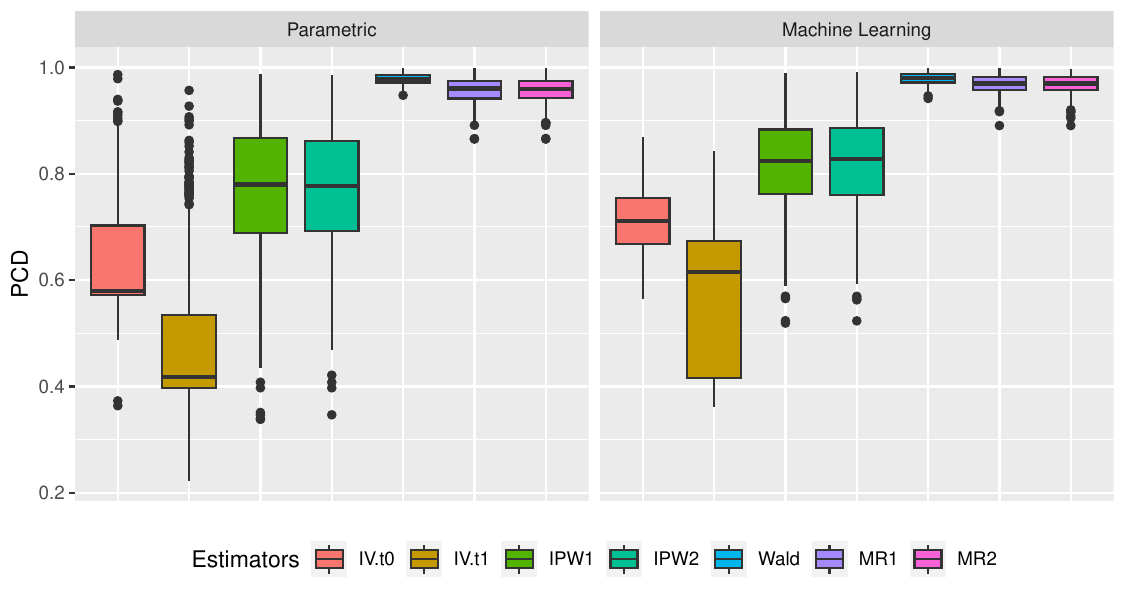}
    \caption{The percentage of correct decisions (PCD) results of the estimated optimal policies, using parametric models (left) or machine learning (right).}
    \label{fig:main}
\end{figure}

Figure~\ref{fig:main} reports the main simulation results from $500$ Monte Carlo replications. In both scenarios, the two standard IV estimators fail to learn the optimal policy, due to the direct effects of the treatment $A$ on the outcomes $Y_0, Y_1$. The two IPW estimators perform much better, but the variability can be large due to possibly extreme weights. The Wald and multiply robust estimators generally lead to lower variability, and attain superior performance. Additional simulation results are reported in Section~\ref{sm.sec:add.simu} of the Supplementary Material to illustrate how different sample sizes and the strength of the IV affect the performance of the estimated policies. We observe that a stronger strength of IV generally leads to lower variability and better accuracy, and also as sample size increases, our proposed methods have better performance.

\section{Data application}
\label{sec:da}

In this section, we illustrate the use of the instrumented DiD approach for policy learning with a analysis of the Australian Longitudinal Survey (ALS) data. Researchers in labor economics have a longstanding interest in investigating the causal effect of education on earnings in the labor market. \cite{card2001estimating} suggests that the endogeneity of education might partially explain the continuing interest ``in this very difficult task of uncovering the causal effect of education in labor market outcomes", and argues that the effects of education are heterogeneous since the economic benefits are individual-specific. Besides the well acknowledged benefits of personal growth and social good from education, we aim to provide a personalized recommendation on whether an individual should pursuit more education or not, in order to gain higher earnings.

The Australian Longitudinal Survey was conducted annually since 1984. Specifically, we include the 1984 and 1985 waves as cross sectional data in our analysis. The 1984 wave surveyed a sample of $3000$ people aged $15 - 24$, and the 1985 wave consisted of $9000$ interviews with people aged $16 - 25$. The surveys aim mainly at providing data on the dynamics of the youth labour market, and include basic demographic variables, labour market variables, background variables and topics related to the main labour market theme. We follow the guidelines from \cite{su2013local,cai2006functional} and \cite{vella1994gender}, who was among the first researchers extensively working with the ALS data. Finally, our data include $2401$ subjects from the 1984 wave, and $8997$ subjects from the 1985 wave. We consider the following baseline covariates: whether a person is born in Australia, marital status, union membership, government employment, age and work experience. The treatment is the education level, and the outcome is the hourly wage. We use an index of labor market attitudes as the instrumental variable \citep{su2013local}. The details of our analysis are provided in Section~\ref{sm.sec:als.data} of the Supplementary Material.

The nuisance parameters are estimated by posited linear/logistic regression models, and we apply our proposed methods with the same configurations as Section~\ref{sec:simu}. The policy coefficient estimates of all covariates are reported in Table~\ref{tab:als.policy.coef}.

\begin{table}[ht]
    \centering
    \begin{tabular}{cccccccc}
        \hline\hline
        Policies & \texttt{intercept} & \texttt{born\_australia} & \texttt{married} & \texttt{uni\_mem} & \texttt{gov\_emp} & \texttt{age} & \texttt{year\_expe} \\
        \hline
        IV.t0 & $0.4442$ & $-0.4547$ & $0.1311$ & $-0.1179$ & $-0.5181$ & $0.0080$ & $-0.5444$ \\
        IV.t1 & $-0.2518$ & $-0.3103$ & $0.2445$ & $-0.6157$ & $-0.1406$ & $0.2015$ & $-0.5840$ \\
        IPW1  & $-0.4203$ & $-0.0847$ & $0.5454$ & $-0.3941$ & $-0.5690$ & $0.0299$ & $0.1969$ \\
        IPW2  & $-0.2503$ & $-0.0529$ & $0.6051$ & $-0.4384$ & $-0.5801$ & $0.0207$ & $0.1980$ \\
        Wald  & $0.5032$ & $0.3891$ & $0.4738$ & $0.5755$ & $-0.1656$ & $-0.0772$ & $0.0793$ \\
        MR1   & $-0.0513$ & $0.1341$ & $-0.6039$ & $0.4127$ & $0.5861$ & $-0.0226$ & $-0.3168$ \\
        MR2   & $0.5480$ & $-0.3937$ & $-0.4072$ & $0.4393$ & $0.4167$ & $-0.0302$ & $-0.1064$ \\
        \hline\hline
    \end{tabular}
    \caption{Coefficients of estimated optimal policy (normalized with $L_2$ norm $1$). \texttt{born\_australia}: whether a person is born in Australia; \texttt{married}: marital status; \texttt{uni\_mem}: union membership; \texttt{gov\_emp}: government employment; \texttt{age}: age; \texttt{year\_expe}: work experience.}
    \label{tab:als.policy.coef}
\end{table}

The coefficients should be interpreted cautiously. We also find that there exists some discrepancies among the treatment recommendations by our proposed estimators. The Wald and multiply robust estimators usually agree, but the variability of the IPW estimators are a bit large. Due to the potentially different recommendations by different estimated policies, one may conservatively suggest a recommendation by the majority rule, and accordingly obtain an ensemble policy. It is also interesting to construct a decision tree to further explore which covariates indicate which treatment level \citep{qi2023proximal}.

\section{Extension to panel data}
\label{sec:panel}

In this section, we consider extending the instrumented DiD approach to the panel data setup where a random sample from the population is followed up over two time points \citep{abadie2005semiparametric}. The observed data are $O = (X, Z, A_{0}, Y_{0}, A_{1}, Y_{1})$. Let $\delta_{Y,z}(x) = E[Y_1 - Y_0 \mid X = x, Z = z]$, $\delta_{A,z}(x) = E[A_1 - A_0 \mid X = x, Z = z]$, and $\pi_{Z}(x) = Pr(Z = 1 \mid X = x)$. We make the following identification assumptions.

\begin{assumption}\label{asmp:panel}
Suppose the following assumptions hold: {\rm (consistency)} $A_{t} = A_{t}(Z)$ and $Y_{t} = Y_{t}(A_{t})$ for $t = 0, 1$; {\rm (positivity)} $c_3 < \pi_{Z}(x) < 1 - c_3$ for some $0 < c_3 < 1/2$; {\rm (trend relevance)} $E[A_{1}(1) - A_{0}(1) \mid Z = 1, X] \neq E[A_{1}(0) - A_{0}(0) \mid Z = 0, X]$; {\rm (stable treatment effect over time)} $E[Y_{0}(1) - Y_{0}(0) \mid X] = E[Y_{1}(1) - Y_{1}(0) \mid X]$; {\rm (independence \& exclusion restriction)} $Z \perp \{A_{t}(1), A_{t}(0), Y_{t}(1) - Y_{t}(0), Y_{1}(0) - Y_{0}(0) : t = 0, 1\} \mid X$; {\rm (no unmeasured common effect modifier)} $Cov\{A_{t}(1) - A_{t}(0), Y_{t}(1) - Y_{t}(0) \mid X\} = 0$ for $t = 0, 1$.
\end{assumption}

Assumption~\ref{asmp:panel} is the counterpart of Assumptions~\ref{asmp:cons}-\ref{asmp:no.com.modi} for the panel/longitudinal structure. \cite{vo2022structural} use a structural mean model and consider alternative assumptions to the no unmeasured common effect modifier assumption above. In Section~\ref{sm.sec:panel.idid} of the Supplementary Material, we also prove the identification results under the following assumptions that replaces the no unmeasured common effect modifier assumption: (sequential ignorability) $Y_{t}(a) \perp A_{t} \mid U, X, Z$ for $t, a = 0, 1$, and there is no additive interaction of either {\rm (i)} $E[A_{1} - A_{0} \mid X, U, Z = 1] - E[A_{1} - A_{0} \mid X, U, Z = 0] = E[A_{1} - A_{0} \mid X, Z = 1] - E[A_{1} - A_{0} \mid X, Z = 0]$ or {\rm (ii)} $E[Y_{t}(1) - Y_{t}(0) \mid U, X] = E[Y_{t}(1) - Y_{t}(0) \mid X]$ for $t = 0, 1$. The sequential ignorability is intuitive, and commonly assumed in panel/longitudinal data analysis. We note that the no additive interaction assumption implies the no unmeasured common effect modifier assumption.

\begin{theorem}\label{thm:panel.cate}
Under Assumption~\ref{asmp:panel}, the CATE is nonparametrically identified by 
\begin{equation}
\tau (x) = \frac{E[Y_{1} - Y_{0} \mid X = x, Z = 1] - E[Y_{1} - Y_{0} \mid X = x, Z = 0]}{E[A_{1} - A_{0} \mid X = x, Z = 1] - E[A_{1} - A_{0} \mid X = x, Z = 0]},
\end{equation}
and the efficient influence function is
\begin{align*}
\phi_{\rm{panel}} = & \frac{\delta_{Y,1}(x) - \delta_{Y,0}(x)}{\delta_{A,1}(x) - \delta_{A,0}(x)} - \frac{z - \pi_{Z}(x)}{\pi_{Z}(x) (1 - \pi_{Z}(x)) (\delta_{A,1}(x) - \delta_{A,0}(x))^2} \left\{(y_{1} - y_{0})(\delta_{A,1}(x) - \delta_{A,0}(x)) \right. \\
& \left. - (a_{1} - a_{0})(\delta_{Y,1}(x) - \delta_{Y,0}(x)) + \delta_{Y,1}(x) \delta_{A,0}(x) - \delta_{Y,0}(x) \delta_{A,1}(x)\right\} - \tau (x).
\end{align*}
\end{theorem}

\begin{theorem}\label{thm:panel.policy}
Under Assumption~\ref{asmp:panel}, the optimal policy is nonparametrically identified by
\begin{equation}
\arg\max_{\mathcal{D}} E\left[\frac{\delta_{Y,1}(X) - \delta_{Y,0}(X)}{\delta_{A,1}(X) - \delta_{A,0}(X)} d(X)\right] = \arg\max_{\mathcal{D}} E\left[\Delta_{\rm{panel}} (X) d(X)\right],
\end{equation}
where the uncentered efficient influence function $\Delta_{\rm{panel}}$ is
\begin{align*}
\Delta_{\rm{panel}} = & \frac{\delta_{Y,1}(X) - \delta_{Y,0}(X)}{\delta_{A,1}(X) - \delta_{A,0}(X)} - \frac{Z - \pi_{Z}(X)}{\pi_{Z}(X) (1 - \pi_{Z}(X)) (\delta_{A,1}(X) - \delta_{A,0}(X))^2} \left\{(Y_{1} - Y_{0})(\delta_{A,1}(X) - \delta_{A,0}(X)) \right. \\
& \left. - (A_{1} - A_{0})(\delta_{Y,1}(X) - \delta_{Y,0}(X)) + \delta_{Y,1}(X) \delta_{A,0}(X) - \delta_{Y,0}(X) \delta_{A,1}(X)\right\}.
\end{align*}
\end{theorem}

Estimators of optimal policies can be constructed by the empirical versions of equations in Theorem~\ref{thm:panel.policy}, and the cross-fitting procedure can also be applied when using the efficient influence function. Similarly, asymptotic analysis of policy learning as Theorems~\ref{thm:asym.para} and \ref{thm:asym.np} can be established for panel data.

\section{Discussion}
\label{sec:disc}

Similar approaches as the instrumented difference-in-differences design has long been employed by econometricians \citep{duflo2001schooling} and has also been formally considered as fuzzy differences-in-differences by \citet{de2018fuzzy}, where the individuals can switch treatment in only one direction within each treatment group. We refer interested readers to \cite{ye2022instrumented} and its rejoinder for discussions on the differences, and applications in biomedicine and epidemiology.

There are several interesting directions for future research and application. Our approach is the first work to systematically study policy learning under the DiD setting. It may be possible to consider alternative assumptions or structures in DiD design to learn the optimal policy. Our instrumented DiD may also be generalized to multiple time points, continuous time, or continuous IV.

Note that Assumption~\ref{asmp:no.com.modi} can be replaced by the \emph{monotonicity} assumption, i.e. $A_{t}(1) \geq A_{t}(0)$ for $t = 0, 1$ with probability $1$, which identifies the complier treatment effects. Then we can also target complier optimal policies that would optimize the potential outcome among compliers.

\subsubsection*{Acknowledgements}

The authors are grateful for helpful comments and feedback from Julie Josse and Antoine Chambaz.

Pan Zhao is supported in part by the French National Research Agency ANR-16-IDEX-0006. Yifan Cui is supported in part by the National Natural Science Foundation of China and the Open Research Fund Key Laboratory of Advanced Theory and Application in Statistics and Data Science (East China Normal University), Ministry of Education of the People’s Republic of China. The authors are grateful to the OPAL infrastructure from Université Côte d'Azur for providing resources and support.

\bibliographystyle{plainnat}
\bibliography{Bibliography}

\newpage
\appendix

\begin{center}
{\large\bf SUPPLEMENTARY MATERIAL}
\end{center}

\section{Directed acyclic graphs}
\label{sm.sec:dag}

In this section, we present the directed acyclic graphs (DAGs) in Figures~\ref{fig:dag1} and \ref{fig:dag2} illustrating the causal structure of the proposed instrumented DiD. The IV $Z$ is associated with the trend in treatment $A_{1} - A_{0}$, is independent of the unmeasured confounders $U_{0}, U_{1}$, cannot have direct effect on the trend in outcome $Y_{1} - Y_{0}$, and does not modofy the treatment effect. But in comparison to a standard IV, here $Z$ is allowed to have a direct effect on the outcomes $Y_{0}, Y_{1}$, as illustrated by the edges $Z \to Y_{0}$ and $Z \to Y_{1}$ in Figure~\ref{fig:dag2}.

\begin{figure}[h]
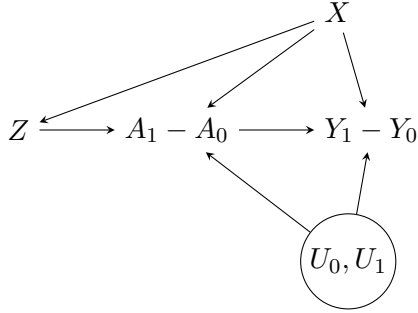

\centering

\tikz{
\node (z) {$Z$};
\node (a) [right = of z] {$A_{1} - A_{0}$};
\node (y) [right = of a] {$Y_{1} - Y_{0}$};
\node (x) [above right = of a] {$X$};
\node[hidden] (u) [below right = of a] {$U_{0}, U_{1}$};
\path (z) edge (a);
\path (a) edge (y);
\path (x) edge (z);
\path (x) edge (a);
\path (x) edge (y);
\path (u) edge (a);
\path (u) edge (y);
}

\caption{DAG for instrumented DiD on the trend scale.}
\label{fig:dag1}
\end{figure}

\begin{figure}[h]
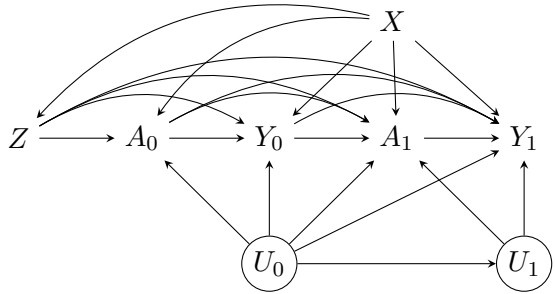

\centering

\tikz{
\node (z) {$Z$};
\node (a0) [right = of z] {$A_{0}$};
\node (y0) [right = of a0] {$Y_{0}$};
\node (a1) [right = of y0] {$A_{1}$};
\node (y1) [right = of a1] {$Y_{1}$};
\node (x) [above right = of y0] {$X$};
\node[hidden] (u0) [below = of y0] {$U_{0}$};
\node[hidden] (u1) [below = of y1] {$U_{1}$};
\path (z) edge (a0);
\path (z) edge[bend left] (y0);
\path (z) edge[bend left] (a1);
\path (z) edge[bend left] (y1);
\path (a0) edge (y0);
\path (a0) edge[bend left] (a1);
\path (a0) edge[bend left] (y1);
\path (y0) edge (a1);
\path (y0) edge[bend left] (y1);
\path (a1) edge (y1);
\path (u0) edge (u1);
\path (u0) edge (a0);
\path (u0) edge (y0);
\path (u0) edge (a1);
\path (u0) edge (y1);
\path (u1) edge (a1);
\path (u1) edge (y1);
\path (x) edge[bend right] (z);
\path (x) edge[bend right] (a0);
\path (x) edge (y0);
\path (x) edge (a1);
\path (x) edge (y1);
}

\caption{DAG for instrumented DiD over two time points.}
\label{fig:dag2}
\end{figure}

\section{Proof of Theorem~\ref{thm:wald}}

In this section, we provide a proof of Theorem~\ref{thm:wald} for completeness. Similar proof can be found at \cite{ye2022instrumented}.

We first note that
\begin{align*}
&\delta_{Y}(X) = \mu_{Y}(1, 1, X) - \mu_{Y}(0, 1, X) - \mu_{Y}(1, 0, X) + \mu_{Y}(0, 0, X) \\
&= \sum_{z=0,1} (2z - 1) (E[Y \mid T = 1, Z = z, X] - E[Y \mid T = 0, Z = z, X]) \\
&= \sum_{z=0,1} (2z - 1) (E[Y_{1}(A_{1}(z)) \mid T = 1, Z = z, X] - E[Y_{0}(A_{0}(z)) \mid T = 0, Z = z, X]) \\
&= \sum_{z=0,1} (2z - 1) (E[Y_{1}(A_{1}(z)) \mid Z = z, X] - E[Y_{0}(A_{0}(z)) \mid Z = z, X]) \\
&= \sum_{z=0,1} (2z - 1) E[Y_{1}(A_{1}(z)) - Y_{0}(A_{0}(z)) \mid Z = z, X] \\
&= \sum_{z=0,1} (2z - 1) E[A_{1}(z) Y_{1}(1) + (1 - A_{1}(z)) Y_{1}(0) - A_{0}(z) Y_{0}(1) - (1 - A_{0}(z)) Y_{0}(0) \mid Z = z, X] \\
&= \sum_{z=0,1} (2z - 1) E[A_{1}(z) (Y_{1}(1) - Y_{1}(0)) - A_{0}(z) (Y_{0}(1) - Y_{0}(0)) + Y_{1}(0) - Y_{0}(0) \mid Z = z, X] \\
&= \sum_{z=0,1} (2z - 1) (E[A_{1}(z) (Y_{1}(1) - Y_{1}(0)) \mid X] - E[A_{0}(z) (Y_{0}(1) - Y_{0}(0)) \mid X] + E[Y_{1}(0) - Y_{0}(0) \mid X]) \\
&= E[(A_{1}(1) - A_{1}(0))(Y_{1}(1) - Y_{1}(0)) \mid X] - E[(A_{0}(1) - A_{0}(0))(Y_{0}(1) - Y_{0}(0)) \mid X] \\
&= E[A_{1}(1) - A_{1}(0) \mid X] E[Y_{1}(1) - Y_{1}(0) \mid X] - E[A_{0}(1) - A_{0}(0) \mid X] E[Y_{0}(1) - Y_{0}(0) \mid X] \\
&= E[A_{1}(1) - A_{1}(0) - A_{0}(1) + A_{0}(0) \mid X] \tau (X).
\end{align*}

Then note that
\begin{align*}
&\delta_{A}(X) = \mu_{A}(1, 1, X) - \mu_{A}(0, 1, X) - \mu_{A}(1, 0, X) + \mu_{A}(0, 0, X) \\
&= \sum_{z=0,1} (2z - 1) (E[A \mid T = 1, Z = z, X] - E[A \mid T = 0, Z = z, X]) \\
&= \sum_{z=0,1} (2z - 1) (E[A_{1}(z) \mid T = 1, Z = z, X] - E[A_{0}(z) \mid T = 0, Z = z, X]) \\
&= E[A_{1}(1) - A_{1}(0) - A_{0}(1) + A_{0}(0) \mid X].
\end{align*}

Hence we have that $\delta_{Y}(X) = \delta_{A}(X) \tau (X)$. That is, the CATE $\tau (X)$ can be identified by $\delta_{Y}(X) / \delta_{A}(X)$. It follows that the optimal policy is nonparametrically identified by
\begin{equation*}
\arg\max_{d \in \mathcal{D}} E[\tau(X)d(X)] = \arg\max_{d \in \mathcal{D}} E\left[\frac{\delta_{Y}(X)}{\delta_{A}(X)} d(X)\right],
\end{equation*}
which completes the proof.

\section{Proof of Theorem~\ref{thm:ipw1}}
\label{sm.sec:proof.ipw1}

In this section, we prove our first novel identification results of the optimal policy.

First we note that 
\begingroup
\allowdisplaybreaks
\begin{align*}
& E\left[\frac{(2Z - 1)(2T - 1)(2A - 1) Y I\{A = d(X)\}}{\pi(T, Z, X) \delta_{A}(X)}\right] \\
& = E\left[\sum_{a = 0, 1} \frac{(2Z - 1)(2T - 1)(2a - 1) Y_{T}(a) I\{A = a\} I\{d(X) = a\}}{\pi(T, Z, X) \delta_{A}(X)}\right] \\
& = E\left[\sum_{a = 0, 1} \frac{(2Z - 1)(2T - 1)(2a - 1) E[Y_{T}(a) \mid X, U] I\{A = a\} I\{d(X) = a\}}{\pi(T, Z, X) \delta_{A}(X)}\right] \\
& = E\left[\sum_{a = 0, 1} \frac{(2Z - 1)(2T - 1)(2a - 1) E[Y_{T}(a) \mid X, U] Pr(A = a \mid X, U, T, Z) I\{d(X) = a\}}{\pi(T, Z, X) \delta_{A}(X)}\right] \\
& = E\left[\frac{Pr(A = 1 \mid X, U, T = 1, Z = 1) I\{d(X) = 1\} E[Y_{1}(1) \mid X, U]}{\delta_{A}(X)}\right] \\
& \quad - E\left[\frac{Pr(A = 1 \mid X, U, T = 0, Z = 1) I\{d(X) = 1\} E[Y_{0}(1) \mid X, U]}{\delta_{A}(X)}\right] \\
& \quad - E\left[\frac{Pr(A = 1 \mid X, U, T = 1, Z = 0) I\{d(X) = 1\} E[Y_{1}(1) \mid X, U]}{\delta_{A}(X)}\right] \\
& \quad + E\left[\frac{Pr(A = 1 \mid X, U, T = 0, Z = 0) I\{d(X) = 1\} E[Y_{0}(1) \mid X, U]}{\delta_{A}(X)}\right] \\
& \quad - E\left[\frac{Pr(A = 0 \mid X, U, T = 1, Z = 1) I\{d(X) = 0\} E[Y_{1}(0) \mid X, U]}{\delta_{A}(X)}\right] \\
& \quad + E\left[\frac{Pr(A = 0 \mid X, U, T = 0, Z = 1) I\{d(X) = 0\} E[Y_{0}(0) \mid X, U]}{\delta_{A}(X)}\right] \\
& \quad + E\left[\frac{Pr(A = 0 \mid X, U, T = 1, Z = 0) I\{d(X) = 0\} E[Y_{1}(0) \mid X, U]}{\delta_{A}(X)}\right] \\
& \quad - E\left[\frac{Pr(A = 0 \mid X, U, T = 0, Z = 0) I\{d(X) = 0\} E[Y_{0}(0) \mid X, U]}{\delta_{A}(X)}\right] \\
& = E\left[\frac{[Pr(A = 1 \mid X, U, T = 1, Z = 1) - Pr(A = 1 \mid X, U, T = 1, Z = 0)] I\{d(X) = 1\} E[Y_{1}(1) \mid X, U]}{\delta_{A}(X)}\right] \\
& \quad + E\left[\frac{[Pr(A = 1 \mid X, U, T = 1, Z = 1) - Pr(A = 1 \mid X, U, T = 1, Z = 0)] I\{d(X) = 0\} E[Y_{1}(0) \mid X, U]}{\delta_{A}(X)}\right] \\
& \quad - E\left[\frac{[Pr(A = 1 \mid X, U, T = 0, Z = 1) - Pr(A = 1 \mid X, U, T = 0, Z = 0)] I\{d(X) = 1\} E[Y_{0}(1) \mid X, U]}{\delta_{A}(X)}\right] \\
& \quad - E\left[\frac{[Pr(A = 1 \mid X, U, T = 0, Z = 1) - Pr(A = 1 \mid X, U, T = 0, Z = 0)] I\{d(X) = 0\} E[Y_{0}(0) \mid X, U]}{\delta_{A}(X)}\right].
\end{align*}
\endgroup

Since we have that for $t = 0, 1$,
\begin{equation*}
\begin{split}
& I\{d(X) = 1\} E[Y_{t}(1) \mid X, U] + I\{d(X) = 0\} E[Y_{t}(0) \mid X, U] \\
& = d(X) (E[Y_{t}(1) \mid X, U] - E[Y_{t}(0) \mid X, U]) + E[Y_{t}(0) \mid X, U] \\
& = d(X) \tau(X) + E[Y_{t}(0) \mid X, U],
\end{split}
\end{equation*}
we continue by Assumption~\ref{asmp:no.com.modi} that
\begin{align*}
& E\left[\frac{(2Z - 1)(2T - 1)(2A - 1) Y I\{A = d(X)\}}{\pi(T, Z, X) \delta_{A}(X)}\right] \\
& = E[d(X) \tau(X)] + E[\nu (X,U)],
\end{align*}
where the second term $E[\nu (X,U)]$ does not depend on the policy $d$. That is, 
\begin{equation*}
\arg\max_{d \in \mathcal{D}} E\left[\frac{(2Z - 1)(2T - 1)(2A - 1) Y I\{A = d(X)\}}{\pi(T, Z, X) \delta_{A}(X)}\right] = \arg\max_{d \in \mathcal{D}} E[\tau(X)d(X)],
\end{equation*}
which completes the proof.

\section{Proof of Theorem~\ref{thm:ipw2}}

In this section, we prove our second novel identification results of the optimal policy.

First we note that 
\begingroup
\allowdisplaybreaks
\begin{align*}
& E\left[\frac{(2T - 1) Y I\{Z = d(X)\}}{\pi(T, Z, X) \delta_{A}(X)}\right] \\
& = E\left[\sum_{a = 0, 1} \frac{(2T - 1) I\{Z = d(X)\} Y_{T}(a) I\{A = a\}}{\pi(T, Z, X) \delta_{A}(X)}\right] \\
& = E\left[\sum_{a = 0, 1} \frac{(2T - 1) I\{Z = d(X)\} E[Y_{T}(a) \mid X, U] Pr(A = a \mid X, U, T, Z)}{\pi(T, Z, X) \delta_{A}(X)}\right] \\
& = E\left[\frac{Pr(A = 1 \mid X, U, T = 1, Z = 1) I\{d(X) = 1\} E[Y_{1}(1) \mid X, U]}{\delta_{A}(X)}\right] \\
& \quad - E\left[\frac{Pr(A = 1 \mid X, U, T = 0, Z = 1) I\{d(X) = 1\} E[Y_{0}(1) \mid X, U]}{\delta_{A}(X)}\right] \\
& \quad + E\left[\frac{Pr(A = 1 \mid X, U, T = 1, Z = 0) I\{d(X) = 0\} E[Y_{1}(1) \mid X, U]}{\delta_{A}(X)}\right] \\
& \quad - E\left[\frac{Pr(A = 1 \mid X, U, T = 0, Z = 0) I\{d(X) = 0\} E[Y_{0}(1) \mid X, U]}{\delta_{A}(X)}\right] \\
& \quad + E\left[\frac{Pr(A = 0 \mid X, U, T = 1, Z = 1) I\{d(X) = 1\} E[Y_{1}(0) \mid X, U]}{\delta_{A}(X)}\right] \\
& \quad - E\left[\frac{Pr(A = 0 \mid X, U, T = 0, Z = 1) I\{d(X) = 1\} E[Y_{0}(0) \mid X, U]}{\delta_{A}(X)}\right] \\
& \quad + E\left[\frac{Pr(A = 0 \mid X, U, T = 1, Z = 0) I\{d(X) = 0\} E[Y_{1}(0) \mid X, U]}{\delta_{A}(X)}\right] \\
& \quad - E\left[\frac{Pr(A = 0 \mid X, U, T = 0, Z = 0) I\{d(X) = 0\} E[Y_{0}(0) \mid X, U]}{\delta_{A}(X)}\right] \\
& = E\left[\frac{[Pr(A = 1 \mid X, U, T = 1, Z = 1) - Pr(A = 1 \mid X, U, T = 1, Z = 0)] I\{d(X) = 1\} E[Y_{1}(1) \mid X, U]}{\delta_{A}(X)}\right] \\
& \quad + E\left[\frac{[Pr(A = 1 \mid X, U, T = 1, Z = 1) - Pr(A = 1 \mid X, U, T = 1, Z = 0)] I\{d(X) = 0\} E[Y_{1}(0) \mid X, U]}{\delta_{A}(X)}\right] \\
& \quad - E\left[\frac{[Pr(A = 1 \mid X, U, T = 0, Z = 1) - Pr(A = 1 \mid X, U, T = 0, Z = 0)] I\{d(X) = 1\} E[Y_{0}(1) \mid X, U]}{\delta_{A}(X)}\right] \\
& \quad - E\left[\frac{[Pr(A = 1 \mid X, U, T = 0, Z = 1) - Pr(A = 1 \mid X, U, T = 0, Z = 0)] I\{d(X) = 0\} E[Y_{0}(0) \mid X, U]}{\delta_{A}(X)}\right] \\
& \quad + E\left[\frac{Pr(A = 1 \mid X, U, T = 1, Z = 0) E[Y_{1}(1) \mid X, U] + Pr(A = 0 \mid X, U, T = 1, Z = 1) E[Y_{1}(0) \mid X, U]}{\delta_{A}(X)}\right] \\
& \quad - E\left[\frac{Pr(A = 1 \mid X, U, T = 0, Z = 0) E[Y_{0}(1) \mid X, U] + Pr(A = 0 \mid X, U, T = 0, Z = 1) E[Y_{0}(0) \mid X, U]}{\delta_{A}(X)}\right].
\end{align*}
\endgroup

Then by the same arguments as in Section~\ref{sm.sec:proof.ipw1}, we have that 
\begin{equation*}
E\left[\frac{(2T - 1) Y I\{Z = d(X)\}}{\pi(T, Z, X) \delta_{A}(X)}\right] = E[d(X) \tau(X)] + E[\tilde{\nu} (X,U)],
\end{equation*}
where the second term does not depend on the policy $d$. That is, 
\begin{equation*}
\arg\max_{d \in \mathcal{D}} E\left[\frac{(2T - 1) Y I\{Z = d(X)\}}{\pi(T, Z, X) \delta_{A}(X)}\right] = \arg\max_{d \in \mathcal{D}} E[\tau(X)d(X)],
\end{equation*}
which completes the proof.

\section{Proof of Theorem~\ref{thm:mr}}

In this section, we prove our identification results of the optimal policy using the efficient influence functions.

First we note that
\begin{align*}
&E[W_{1} I\{A = d(X)\}] \\
&= \frac{1}{2} E[W_{1} (2I\{A = d(X)\} - 1)] + \frac{1}{2} E[W_{1}] \\
&= \frac{1}{2} E[W_{1} (2A - 1) (2d(X) - 1)] + \frac{1}{2} E[W_{1}] \\
&= \frac{1}{2} E[\Delta(O) (2d(X) - 1)] + \frac{1}{2} E[W_{1}] \\
&= E[\Delta(O) d(X)] + \frac{1}{2} E[W_{1} - \Delta(O)] \\
&= E[\tau(X) d(X)] + \frac{1}{2} E[W_{1} - \Delta(O)],
\end{align*}
where the last equality holds under the union model $\mathcal{M}_1 \cup \mathcal{M}_2 \cup \mathcal{M}_3$. The proof of the multiple robustness is omitted since it simply follows the same arguments of Theorem 1 in \cite{ye2022instrumented}.

We also note that
\begin{align*}
&E[W_{2} I\{Z = d(X)\}] \\
&= \frac{1}{2} E[W_{2} (2I\{Z = d(X)\} - 1)] + \frac{1}{2} E[W_{2}] \\
&= \frac{1}{2} E[W_{2} (2Z - 1) (2d(X) - 1)] + \frac{1}{2} E[W_{2}] \\
&= \frac{1}{2} E[\Delta(O) (2d(X) - 1)] + \frac{1}{2} E[W_{2}] \\
&= E[\Delta(O) d(X)] + \frac{1}{2} E[W_{2} - \Delta(O)] \\
&= E[\tau(X) d(X)] + \frac{1}{2} E[W_{2} - \Delta(O)],
\end{align*}
where the last equality holds under the union model $\mathcal{M}_1 \cup \mathcal{M}_2 \cup \mathcal{M}_3$.

\section{A locally efficient and multiply robust estimator}
\label{sm.sec:semi.eff}

In this section, we present the semiparametric efficiency results for our proposed IPW formula:
\begin{equation*}
\Psi(P) = E\left[\frac{(2Z - 1)(2T - 1)(2A - 1) Y I\{A = d(X)\}}{\pi(T, Z, X) \delta_{A}(X)}\right].
\end{equation*}

We first characterize the efficient influence function, and then propose the multiply robust estimator.

\begin{theorem}\label{thm:value.eif}
The efficient influence function of $\Psi(P)$ is
\begin{align*}
\phi_{P} &= \frac{(2Z - 1)(2T - 1)(2A - 1) Y I\{A = d(X)\}}{\pi(T, Z, X) \delta_{A}(X)} \\
&\qquad - \frac{(2Z - 1)(2T - 1)E[(2A - 1) Y I\{A = d(X)\} \mid T, Z, X]}{\pi(T, Z, X) \delta_{A}(X)} + \gamma(X) \\
&\qquad - \frac{(2Z - 1)(2T - 1)(A - \mu_{A}(T, Z, X))\gamma(X)}{\pi(T, Z, X) \delta_{A}(X)} - \Psi(P),
\end{align*}
where $\gamma(x) = \sum_{t, z} (2z - 1)(2t - 1) E[(2A - 1) Y I\{A = d(X)\} \mid T = t, Z = z, X = x] / \delta_{A}(x)$.
\end{theorem}

By Theorem~\ref{thm:value.eif}, we conclude that the optimal policy is nonparametrically identified by $\arg\max_{\mathcal{D}} \psi_{P}$, where 
\begin{align*}
\psi_{P} &= E\left[\frac{(2Z - 1)(2T - 1)(2A - 1) Y I\{A = d(X)\}}{\pi(T, Z, X) \delta_{A}(X)} \right. \\
&\quad \left. - \frac{(2Z - 1)(2T - 1)E[(2A - 1) Y I\{A = d(X)\} \mid T, Z, X]}{\pi(T, Z, X) \delta_{A}(X)} + \gamma(X) \right. \\
&\quad \left. - \frac{(2Z - 1)(2T - 1)(A - \mu_{A}(T, Z, X))\gamma(X)}{\pi(T, Z, X) \delta_{A}(X)}\right].    
\end{align*}

In Theorem~\ref{thm:value.mr}, we show the multiple robustness of the above formula under models: \\
$\Tilde{\mathcal{M}}_1$: models for $\pi(t, z, x)$ and $\delta_{A}(x)$ are correct; \\
$\Tilde{\mathcal{M}}_2$: models for $\pi(t, z, x)$ and $\gamma(x)$ are correct; \\
$\Tilde{\mathcal{M}}_3$: models for $\mu_{A}(t, z, x)$, $\gamma(x)$ and $\nu (t, z, x)$ are correct, where $\nu (t, z, x) = E[(2A - 1) Y I\{A = d(X)\} \mid T = t, Z = z, X = x]$.

\begin{theorem}\label{thm:value.mr}
Under standard regularity conditions, we have that
\begin{align*}
P_{n} \psi (\hat{P}) = P_{n} & \left[\frac{(2Z - 1)(2T - 1)(2A - 1) Y I\{A = d(X)\}}{\hat{\pi}(T, Z, X) \hat{\delta}_{A}(X)} \right. \\
&\quad \left. - \frac{(2Z - 1)(2T - 1) \hat{E}[(2A - 1) Y I\{A = d(X)\} \mid T, Z, X]}{\hat{\pi}(T, Z, X) \hat{\delta}_{A}(X)} + \hat{\gamma}(X) \right. \\
&\quad \left. - \frac{(2Z - 1)(2T - 1)(A - \hat{\mu}_{A}(T, Z, X)) \hat{\gamma}(X)}{\hat{\pi}(T, Z, X) \hat{\delta}_{A}(X)} \right]
\end{align*}
is a consistent and asymptotically normal estimator of $\Psi(P)$ under the union model $\Tilde{\mathcal{M}}_1 \cup \Tilde{\mathcal{M}}_2 \cup \Tilde{\mathcal{M}}_3$. Furthermore, it is locally efficient under the intersection model $\Tilde{\mathcal{M}}_1 \cap \Tilde{\mathcal{M}}_2 \cap \Tilde{\mathcal{M}}_3$.
\end{theorem}

Despite the fact that we characterize the efficient influence function and propose a multiply robust estimator, note that it is not straightforward to posit models for $\gamma(x)$ and $\nu(t,z,x)$.

\section{Proof of Theorem~\ref{thm:value.eif} and Theorem~\ref{thm:value.mr}}
\label{sm.sec:semi.eff.proof}

We first prove Theorem~\ref{thm:value.eif} by deriving the efficient influence function.

For a given distribution $P$ in the nonparametric statistical model $\mathcal{M}$, we let $p$ denote the density of $P$ with respect to some dominating measure $\nu$. For all bounded $h \in L_{2}(P)$, define the parametric submodel $p_{\epsilon} = (1 + \epsilon h) p$, which is valid for small enough $\epsilon$ and has score $h$ at $\epsilon = 0$.

We study the following statistical functional
\begin{equation*}
\Psi (P) = E_{P} \left[\frac{(2Z - 1)(2T - 1)(2A - 1) Y I\{A = d(X)\}}{\pi(T, Z, X) \delta_{A}(X)}\right],
\end{equation*}
and would establish that $\Psi (P)$ is pathwise differentiable with respect to $\mathcal{M}$ at $P$ with efficient influence function $\phi_{P}$ if we have that for any $P \in \mathcal{M}$
\begin{equation*}
\frac{\partial}{\partial\epsilon} \Psi (P_{\epsilon}) \bigg|_{\epsilon = 0} = \int \phi_{P} (o) h(o) dP(o).
\end{equation*}

We denote $\pi_{\epsilon}(t, z, x) = E_{P_{\epsilon}} [I\{T = t, Z = z\} \mid X = x]$, $\delta_{A,\epsilon}(x) = \mu_{A,\epsilon}(1, 1, x) - \mu_{A,\epsilon}(0, 1, x) - \mu_{A,\epsilon}(1, 0, x) + \mu_{A,\epsilon}(0, 0, x)$, $\mu_{A,\epsilon}(t, z, x) = E_{P_{\epsilon}}[A \mid T = t, Z = z, X = x]$, $S = \partial \log p_{\epsilon} / \partial\epsilon$, and compute
\begin{align*}
\frac{\partial}{\partial\epsilon} \Psi (P_{\epsilon}) \bigg|_{\epsilon = 0} &= \frac{\partial}{\partial\epsilon} E_{P_{\epsilon}} \left[\frac{(2Z - 1)(2T - 1)(2A - 1) Y I\{A = d(X)\}}{\pi_{\epsilon}(T, Z, X) \delta_{A,\epsilon}(X)}\right] \bigg|_{\epsilon = 0} \\
&= \frac{\partial}{\partial\epsilon} E_{P} \left[(1 + \epsilon S) \frac{(2Z - 1)(2T - 1)(2A - 1) Y I\{A = d(X)\}}{\pi_{\epsilon}(T, Z, X) \delta_{A,\epsilon}(X)}\right] \bigg|_{\epsilon = 0} \\
&= E_{P} \left[S \frac{(2Z - 1)(2T - 1)(2A - 1) Y I\{A = d(X)\}}{\pi(T, Z, X) \delta_{A}(X)}\right] \\
&\quad - E_{P} \left[\frac{(2Z - 1)(2T - 1)(2A - 1) Y I\{A = d(X)\}}{\pi^{2}(T, Z, X) \delta_{A}^{2}(X)} \left(\delta_{A}(X) \frac{\partial}{\partial\epsilon} \pi_{\epsilon}(T, Z, X) \bigg|_{\epsilon = 0} \right.\right. \\
&\qquad\qquad \left.\left. + \pi(T, Z, X) \frac{\partial}{\partial\epsilon} \delta_{A,\epsilon}(X) \bigg|_{\epsilon = 0} \right)\right].
\end{align*}

Then we need to compute
\begin{align*}
\frac{\partial}{\partial\epsilon} \pi_{\epsilon}(t, z, X) \bigg|_{\epsilon = 0} &= \frac{\partial}{\partial\epsilon} E_{P_{\epsilon}} [I\{T = t, Z = z\} \mid X] \bigg|_{\epsilon = 0} \\
&= \frac{\partial}{\partial\epsilon} \frac{\pi(t, z, X) + \epsilon E_{P}[S I\{T = t, Z = z\} \mid X]}{1 + \epsilon E_{P}[S \mid X]} \bigg|_{\epsilon = 0} \\
&= E_{P}[S I\{T = t, Z = z\} \mid X] - \pi(t, z, X) E_{P}[S \mid X],
\end{align*}
\begin{equation*}
\frac{\partial}{\partial\epsilon} \delta_{A,\epsilon}(X) \bigg|_{\epsilon = 0} = \frac{\partial}{\partial\epsilon} \{\mu_{A,\epsilon}(1, 1, X) - \mu_{A,\epsilon}(0, 1, X) - \mu_{A,\epsilon}(1, 0, X) + \mu_{A,\epsilon}(0, 0, X)\} \bigg|_{\epsilon = 0},
\end{equation*}
and 
\begin{align*}
\frac{\partial}{\partial\epsilon} \mu_{A,\epsilon}(t, z, X) \bigg|_{\epsilon = 0} &= \frac{\partial}{\partial\epsilon} E_{P_{\epsilon}}[A \mid T = t, Z = z, X] \bigg|_{\epsilon = 0} \\
&= \frac{\partial}{\partial\epsilon} \frac{\mu_{A}(t, z, X) + \epsilon E_{P}[S A \mid T = t, Z = z, X]}{1 + \epsilon E_{P}[S \mid T = t, Z = z, X]} \bigg|_{\epsilon = 0} \\
&= E_{P}[S A \mid T = t, Z = z, X] - \mu_{A}(t, z, X) E_{P}[S \mid T = t, Z = z, X] \\
&= E_{P}\left[S \frac{(A - \mu_{A}(t, z, X)) I\{T = t, Z = z\}}{\pi (t,z,X)} \mid X \right].
\end{align*}

In summary, we obtain the efficient influence function
\begin{align*}
\phi_{P} &= \frac{(2Z - 1)(2T - 1)(2A - 1) Y I\{A = d(X)\}}{\pi(T, Z, X) \delta_{A}(X)} \\
&\quad - \frac{(2Z - 1)(2T - 1)E[(2A - 1) Y I\{A = d(X)\} \mid T, Z, X]}{\pi(T, Z, X) \delta_{A}(X)} + \gamma(X) \\
&\quad - \frac{(2Z - 1)(2T - 1)(A - \mu_{A}(T, Z, X))\gamma(X)}{\pi(T, Z, X) \delta_{A}(X)} - \Psi(P),
\end{align*}
which completes the proof of Theorem~\ref{thm:value.eif}.

Next, we prove Theorem~\ref{thm:value.mr} by verifying the multiple robustness property.

We first note the facts that $\mu_{A}(T, Z, X) = \mu_{A}(0, 0, x) + Z (\mu_{A}(0, 1, x) - \mu_{A}(0, 0, x)) + T (\mu_{A}(1, 0, x) - \mu_{A}(0, 0, x)) + T Z \delta_{A}(X)$, $\nu(T, Z, X) = \nu(0, 0, x) + Z (\nu(0, 1, x) - \nu(0, 0, x)) + T (\nu(1, 0, x) - \nu(0, 0, x)) + T Z \delta_{A}(X)$, $E[(2Z-1)(2T-1)/\pi(T,Z,X) \mid T, X] = E[(2Z-1)(2T-1)/\pi(T,Z,X) \mid Z, X] = 0$, and $E[\gamma (X)] = \Psi(P)$.

If $\Tilde{\mathcal{M}}_1$ is correctly specified, we have that 
\begin{align*}
E[\phi_{P} (O)] &= E\left[\frac{(2Z - 1)(2T - 1)(A - \mu_{A}(T, Z, X))\gamma(X)}{\pi(T, Z, X) \delta_{A}(X)}\right] \\
&= E\left[\frac{(2Z - 1)(2T - 1)\gamma(X)}{\pi(T, Z, X) \delta_{A}(X)}(A - \mu_{A}(T, Z, X))\right] = 0.
\end{align*}

If $\Tilde{\mathcal{M}}_2$ is correctly specified, we have that
\begin{align*}
E[\phi_{P} (O)] &= E\left[\frac{(2Z - 1)(2T - 1)(2A - 1) Y I\{A = d(X)\}}{\pi(T, Z, X) \delta_{A}(X)} \right.\\
&\qquad\left. - \frac{(2Z - 1)(2T - 1)E[(2A - 1) Y I\{A = d(X)\} \mid T, Z, X]}{\pi(T, Z, X) \delta_{A}(X)} \right.\\
&\qquad\left. - \frac{(2Z - 1)(2T - 1)(A - \mu_{A}(T, Z, X))\gamma(X)}{\pi(T, Z, X) \delta_{A}(X)} \right] = 0.
\end{align*}

If $\Tilde{\mathcal{M}}_3$ is correctly specified, we have that
\begin{align*}
E[\phi_{P} (O)] &= E\left[\frac{(2Z - 1)(2T - 1)(2A - 1) Y I\{A = d(X)\}}{\pi(T, Z, X) \delta_{A}(X)} \right.\\
&\qquad\left. - \frac{(2Z - 1)(2T - 1)E[(2A - 1) Y I\{A = d(X)\} \mid T, Z, X]}{\pi(T, Z, X) \delta_{A}(X)}\right] = 0,
\end{align*}
which completes the proof.

\section{Proof of Theorem~\ref{thm:asym.para}}

We study the following maximization problem:
\begin{align*}
\hat{\eta} = \arg\max_{\eta \in \mathbb{H}} \frac{1}{n}\sum_{i=1}^{n} & \left(\frac{\delta_{Y}(X_i;\hat{\beta})}{\delta_{A}(X_i;\hat{\alpha})} + \frac{(2Z_i - 1)(2T_i - 1)}{\pi(T_i, Z_i, X_i;\hat{\theta}) \delta_{A}(X_i;\hat{\alpha})} \left\{Y_i - \mu_{Y}(T_i, Z_i, X_i;\hat{\beta}) \right.\right.\\ 
&\qquad \left.\left. - \frac{\delta_{Y}(X_i;\hat{\beta})}{\delta_{A}(X_i;\hat{\alpha})} (A_{i} - \mu_{A}(T_i, Z_i, X_i;\hat{\alpha})) \right\} \right)d(X_i;\eta),
\end{align*}
where $\hat{\alpha}$, $\hat{\beta}$ and $\hat{\theta}$ are estimated by posited parametric models. We let $\hat{M}(\eta)$ denote the estimated objective function above, i.e. $\hat{\eta} = \arg\max_{\eta \in \mathbb{H}} \hat{M}(\eta)$.

Under standard regularity conditions, we have that
\begin{align*}
\sqrt{n} (\hat{\alpha} - \alpha^\ast) &= \frac{1}{\sqrt{n}} \sum_{i=1}^n \phi_{\alpha, i} + o_p(1), \\ \sqrt{n} (\hat{\beta} - \beta^\ast) &= \frac{1}{\sqrt{n}} \sum_{i=1}^n \phi_{\beta, i} + o_p(1), \\ \sqrt{n} (\hat{\theta} - \theta^\ast) &= \frac{1}{\sqrt{n}} \sum_{i=1}^n \phi_{\theta, i} + o_p(1),    
\end{align*}
where $\alpha^\ast$, $\beta^\ast$ and $\theta^\ast$ are the probability limits, $\phi_{\alpha}$, $\phi_{\beta}$ and $\phi_{\theta}$ are the influence functions.

Now we start our proof which has three main parts as follows.

{\bf PART 1.} First we note that, by the multiple robustness property, the strong law of large numbers and uniform consistency, $\hat{M}(\eta) = M(\eta) + o_p(1)$. 

We denote
\begin{align*}
M_{n}^\ast (\eta) = \frac{1}{n}\sum_{i=1}^{n} & \left(\frac{\delta_{Y}(X_i;\beta^\ast)}{\delta_{A}(X_i;\alpha^\ast)} + \frac{(2Z_i - 1)(2T_i - 1)}{\pi(T_i, Z_i, X_i;\theta^\ast) \delta_{A}(X_i;\alpha^\ast)} \left\{Y_i - \mu_{Y}(T_i, Z_i, X_i;\beta^\ast) \right.\right.\\ 
&\qquad \left.\left. - \frac{\delta_{Y}(X_i;\beta^\ast)}{\delta_{A}(X_i;\alpha^\ast)} (A_{i} - \mu_{A}(T_i, Z_i, X_i;\alpha^\ast)) \right\} \right)d(X_i;\eta),
\end{align*}
and apply the Taylor expansion on $\hat{M}(\eta)$ at $(\alpha^\ast, \beta^\ast, \theta^\ast)$,
\begin{equation*}
\hat{M}(\eta) = M_{n}^\ast (\eta) + H_{\alpha^\ast}^{\rm T} (\hat{\alpha} - \alpha^\ast) + H_{\beta^\ast}^{\rm T} (\hat{\beta} - \beta^\ast) + H_{\theta^\ast}^{\rm T} (\hat{\theta} - \theta^\ast) + o_p(n^{-1/2}),
\end{equation*}
where $H_{\alpha^\ast} = \lim_{n \to \infty} \partial \hat{M}(\eta) / \partial \alpha |_{\alpha = \alpha^\ast}$, $H_{\beta^\ast} = \lim_{n \to \infty} \partial \hat{M}(\eta) / \partial \beta |_{\beta=\beta^\ast}$, and $H_{\theta^\ast} = \lim_{n \to \infty} \partial \hat{M}(\eta) / \partial \theta |_{\theta=\theta^\ast}$.

Hence, we obtain that 
\begin{equation}\label{eq:asym.dist}
\sqrt{n} \left\{\hat{M}(\eta) - M(\eta)\right\} = \frac{1}{\sqrt{n}} \sum_{i=1}^{n} \left(M_{n}^\ast (\eta) - M(\eta) + H_{\alpha^\ast}^{\rm T} \phi_{\alpha, i} + H_{\beta^\ast}^{\rm T} \phi_{\beta, i} + H_{\theta^\ast}^{\rm T} \phi_{\theta, i}\right) + o_p(1).
\end{equation}

{\bf PART 2.} We prove that $n^{1/3} \|\hat{\eta} - \eta^\ast\|_2 = O_p(1)$.

First we note that, by Condition~\ref{cond.policy} (iii), $M(\eta)$ is twice continuously differentiable at a neighborhood of $\eta^\ast$. In PART 1, we show that $\hat{M}(\eta) = M(\eta) + o_p(1), \forall \eta$. Since $\hat{\eta}$ maximizes $\hat{M}(\eta)$, we have that $\hat{M}(\hat{\eta}) \geq \sup_{\eta} \hat{M}(\eta)$; thus by the Argmax theorem, we obtain that $\hat{\eta} \overset{p}{\to} \eta^\ast$ as $n \to \infty$. 

Then we apply Theorem 14.4 (Rate of convergence) of \cite{kosorok2008introduction} to establish the $n^{-1/3}$ rate of convergence of $\hat{\eta}$, and need to find the suitable rate that satisfies three conditions below.

{\bf Condition 1} For every $\eta$ in a neighborhood of $\eta^\ast$ such that $\|\eta - \eta^\ast\|_2 < \delta$, by Condition~\ref{cond.policy} (iii), we apply the second-order Taylor expansion,
\begin{align*}
M(\eta) - M(\eta^\ast) &= M'(\eta^\ast) \|\eta - \eta^\ast\|_2 + \frac{1}{2} M''(\eta^\ast) \|\eta - \eta^\ast\|_2^2 + o(\|\eta - \eta^\ast\|_2^2) \\
&= \frac{1}{2} S''(\eta^\ast)\|\eta - \eta^\ast\|_2^2 + o(\|\eta - \eta^\ast\|_2^2),
\end{align*}
and as $S''(\eta^\ast) < 0$, there exists $c_0 = -\frac{1}{2} S''(\eta^\ast) > 0$ such that $S(t;\eta) - S(t;\eta^\ast) \leq c_0 \|\eta - \eta^\ast\|_2^2$. 

{\bf Condition 2} For all $n$ large enough and sufficiently small $\delta$, we consider the centered process $\hat{M} - M$, and have that 
\begin{align*}
& E^\ast \left[ \sqrt{n} \sup_{\|\eta - \eta^\ast\|_2 < \delta} \left| \hat{M}(\eta) - M(\eta) - \left\{\hat{M}(\eta^\ast) - M(\eta^\ast)\right\} \right| \right] \\
&= E^\ast \left[ \sqrt{n} \sup_{\|\eta - \eta^\ast\|_2 < \delta} \left| \hat{M}(\eta) - M_n^\ast(\eta) + M_n^\ast(\eta) - M(\eta) - \left\{ \hat{M}(\eta^\ast) - M_n^\ast(\eta^\ast) + M_n^\ast(\eta^\ast) - M(\eta^\ast)\right\} \right| \right] \\
& \leq E^\ast \left[ \sqrt{n} \sup_{\|\eta - \eta^\ast\|_2 < \delta} \left| \hat{M}(\eta) - M_n^\ast(\eta) - \left\{ \hat{M}(\eta^\ast) - M_n^\ast(\eta^\ast) \right\} \right| \right] \\
& \quad + E^\ast \left[ \sqrt{n} \sup_{\|\eta - \eta^\ast\|_2 < \delta} \left| M_n^\ast(\eta) - M(\eta) - \left\{ M_n^\ast(\eta^\ast) - M(\eta^\ast)\right\} \right| \right] \\
&= (I) + (II),
\end{align*}
where $E^\ast (\cdot)$ denote the outer expectation, and we bound $(I)$ and $(II)$ respectively as follows.

{\bf Condition 2.1} To bound $(II)$, we note that
\begin{align*}
M_n^\ast (\eta) - M_n^\ast (\eta^\ast) &= \frac{1}{n} \sum_{i=1}^{n} \Delta^\ast (O_i) (d(X_i;\eta) - d(X_i;\eta\ast)) \\
&= \frac{1}{n} \sum_{i=1}^{n} \Delta^\ast (O_i) (I\{X_{i}^{\rm T} \eta > 0\} - I\{X_{i}^{\rm T} \eta^\ast > 0\}),
\end{align*}
where
\begin{equation*}
\Delta^\ast (o) = \frac{\delta_{Y}(x;\beta^\ast)}{\delta_{A}(x;\alpha^\ast)} + \frac{(2z - 1)(2t - 1)}{\pi(t, z, x;\theta^\ast) \delta_{A}(x;\alpha^\ast)} \left\{y - \mu_{Y}(t, z, x;\beta^\ast) - \frac{\delta_{Y}(x;\beta^\ast)}{\delta_{A}(x;\alpha^\ast)} (a - \mu_{A}(t, z, x;\alpha^\ast)) \right\}.
\end{equation*}

We define a class of functions
\begin{equation*}
\mathcal{F}_\eta^1 (o) = \left\{\Delta^\ast (o) (I\{x^{\rm T} \eta > 0\} - I\{x^{\rm T} \eta^\ast > 0\}) : \|\eta - \eta^\ast\|_2 < \delta \right\},     
\end{equation*}
and let $B_1 = \sup |\Delta^\ast (o)|$. By Assumption~\ref{asmp:posi} and Condition~\ref{cond.policy}, we have that $B_1 < \infty$. 

When $\|\eta -\eta^\ast\|_2 < \delta$, by Condition~\ref{cond.policy} (i), there exists a constant $0 < k_0 < \infty$ such that $|x^{\rm T} (\eta - \eta^\ast)| < k_0 \delta$. Furthermore, we show that $|d(x;\eta) - d(x;\eta^\ast)| = |I\{x^{\rm T} \eta > 0\} - I\{x^{\rm T} \eta^\ast > 0\}| \leq I\{-k_0\delta\leq x^{\rm T}\eta^\ast\leq k_0\delta\}$, by considering the three cases:
\begin{itemize}
\item when $-k_0\delta \leq x^{\rm T}\eta^\ast \leq k_0\delta$, we have $|d(x;\eta) - d(x;\eta^\ast)| \leq 1 = I\{-k_0\delta\leq x^{\rm T}\eta^\ast\leq k_0\delta\}$;
\item when $x^{\rm T}\eta^\ast > k_0 \delta > 0$, we have $x^{\rm T}\eta = x^{\rm T}(\eta - \eta^\ast) + x^{\rm T}\eta^\ast > 0$, so $|d(x;\eta) - d(x;\eta^\ast)| = 0 = I\{-k_0\delta\leq x^{\rm T}\eta^\ast\leq k_0\delta\}$; 
\item when $x^{\rm T}\eta^\ast < -k_0 \delta < 0$, we have $x^{\rm T}\eta = x^{\rm T}(\eta - \eta^\ast) + x^{\rm T}\eta^\ast < 0$, so $|d(x;\eta) - d(x;\eta^\ast)| = 0 = I\{-k_0\delta\leq x^{\rm T}\eta^\ast\leq k_0\delta\}$.
\end{itemize}

Thus we define the envelope of $\mathcal{F}_\eta^1$ as $F_1 = B_1 I\{-k_0 \delta \leq x^{\rm T} \eta^\ast \leq k_0 \delta\}$. By Condition~\ref{cond.policy} (iv), there exists a constant $0 < k_1 < \infty$ such that
\begin{equation*}
\|F_1\|_{P,2} \leq B_1 \sqrt{Pr(-k_0 \delta \leq x^{\rm T} \eta^\ast \leq k_0 \delta)} \leq B_1 \sqrt{2k_0 k_1} \delta^{1/2} < \infty.
\end{equation*}

By Lemma 9.6 and Lemma 9.9 of \cite{kosorok2008introduction}, we have that $\mathcal{F}_\eta^1$, a class of indicator functions, is a Vapnik-Cervonenkis (VC) class with bounded bracketing entropy $J_{[]}^\ast(1,\mathcal{F}_\eta^1) < \infty$.

Next, we note that
\begin{align*}
\mathbb{G}_n \mathcal{F}_\eta^1 &= n^{-1/2} \sum_{i=1}^{n} \left\{\mathcal{F}_\eta^1 (O_i) - E[\mathcal{F}_\eta^1 (O)] \right\} \\
&= \sqrt{n} \left(M_n^\ast(\eta) - M_n^\ast(\eta^\ast) - \left\{M(\eta) - M(\eta^\ast)\right\}\right),
\end{align*}
and by Theorem 11.2 of \cite{kosorok2008introduction}, we obtain that there exists a constant $0 < c_1 < \infty$,
\begin{equation*}
(II) = E^\ast \left[\sup_{\|\eta - \eta^\ast\|_2 < \delta} |\mathbb{G}_n \mathcal{F}_\eta^1|\right] \leq c_1 J_{[]}^\ast(1,\mathcal{F}_\eta^1) \|F_1\|_{P,2} \leq c_1 J_{[]}^\ast(1,\mathcal{F}_\eta^1)B_1\sqrt{2k_0 k_1} \delta^{1/2} = \tilde{c}_1 \delta^{1/2},
\end{equation*}
hence we conclude that $(II) \leq \tilde{c}_1 \delta^{1/2}$, where $\tilde{c}_1 > 0$ is a finite constant.

{\bf Condition 2.2} To bound $(I)$, first we note that
\begin{align*}
& \hat{M}(\eta) - M_n^\ast(\eta) - \{\hat{M}(\eta^\ast) - M_n^\ast(\eta^\ast)\} = \hat{M}(\eta) - \hat{M}(\eta^\ast) - \{M_n^\ast(\eta) - M_n^\ast(\eta^\ast)\} \\
&= \frac{1}{n} \sum_{i=1}^{n} (d(X_i;\eta) - d(X_i;\eta\ast)) (\hat{\Delta}(O_i) - \Delta^\ast (O_i)),
\end{align*}
and then apply the Taylor expansion at $(\alpha^\ast, \beta^\ast, \theta^\ast)$
\begin{equation}\label{eq:bound1}
\begin{split}
& \hat{M}(\eta) - M_n^\ast(\eta) - \{\hat{M}(\eta^\ast) - M_n^\ast(\eta^\ast)\} \\
&= \frac{1}{n} \sum_{i=1}^{n} (d(X_i;\eta) - d(X_i;\eta\ast)) \left\{\left[g_{1}^\ast (O_i) \left(\frac{\partial \delta_{A}(X_i;\alpha^\ast)}{\partial \alpha}\right)^{\rm T} + g_{2}^\ast (O_i) \left(\frac{\partial \mu_{A}(T_i, Z_i, X_i;\alpha^\ast)}{\partial \alpha}\right)^{\rm T}\right](\hat{\alpha} - \alpha^\ast) \right. \\
&\qquad \left. + \left[g_{3}^\ast (O_i) \left(\frac{\partial \delta_{Y}(X_i;\beta^\ast)}{\partial \beta}\right)^{\rm T} + g_{4}^\ast (O_i) \left(\frac{\partial \mu_{Y}(T_i, Z_i, X_i;\beta^\ast)}{\partial \beta}\right)^{\rm T}\right](\hat{\beta} - \beta^\ast) \right. \\
&\qquad \left. + g_{5}^\ast (O_i) \left(\frac{\partial \pi(T_i, Z_i, X_i;\theta^\ast)}{\partial \theta}\right)^{\rm T}(\hat{\theta} - \theta^\ast)\right\} + o_p(n^{-1/2}),
\end{split}
\end{equation}
where
\begin{equation*}
g_{1}^\ast (o) = -\frac{\delta_{Y}(x;\beta^\ast)}{\delta_{A}^{2}(x;\alpha^\ast)} - \frac{(2z - 1)(2t - 1)(y - \mu_{Y}(t, z, x;\beta^\ast))}{\pi(t, z, x;\theta^\ast) \delta_{A}^{2}(x;\alpha^\ast)} + \frac{2(2z - 1)(2t - 1)\delta_{Y}(x;\beta^\ast)}{\pi(t, z, x;\theta^\ast) \delta_{A}^{3}(x;\alpha^\ast)} (a - \mu_{A}(t, z, x;\alpha^\ast)),
\end{equation*}
\begin{equation*}
g_{2}^\ast (o) = \frac{(2z - 1)(2t - 1)\delta_{Y}(x;\beta^\ast)}{\pi(t, z, x;\theta^\ast) \delta_{A}^{2}(x;\alpha^\ast)},
\end{equation*}
\begin{equation*}
g_{3}^\ast (o) = \frac{1}{\delta_{A}^{2}(x;\alpha^\ast)} - \frac{2(2z - 1)(2t - 1)}{\pi(t, z, x;\theta^\ast) \delta_{A}^{2}(x;\alpha^\ast)} (a - \mu_{A}(t, z, x;\alpha^\ast)),
\end{equation*}
\begin{equation*}
g_{4}^\ast (o) = -\frac{(2z - 1)(2t - 1)}{\pi(t, z, x;\theta^\ast) \delta_{A}(x;\alpha^\ast)},
\end{equation*}
\begin{equation*}
g_{5}^\ast (o) = -\frac{(2z - 1)(2t - 1)}{\pi^{2}(t, z, x;\theta^\ast) \delta_{A}(x;\alpha^\ast)} \left\{y - \mu_{Y}(t, z, x;\beta^\ast) - \frac{\delta_{Y}(x;\beta^\ast)}{\delta_{A}(x;\alpha^\ast)} (a - \mu_{A}(t, z, x;\alpha^\ast)) \right\}.
\end{equation*}

Similarly, we define the following classes of functions
\begin{equation*}
\mathcal{F}_\eta^2 (o) = \left\{\left[g_{1}^\ast (o) \left(\frac{\partial \delta_{A}(x;\alpha^\ast)}{\partial \alpha}\right)^{\rm T} + g_{2}^\ast (o) \left(\frac{\partial \mu_{A}(t, z, x;\alpha^\ast)}{\partial \alpha}\right)^{\rm T}\right] (I\{x^{\rm T} \eta > 0\} - I\{x^{\rm T} \eta^\ast > 0\}) : \|\eta - \eta^\ast\|_2 < \delta \right\},     
\end{equation*}
\begin{equation*}
\mathcal{F}_\eta^3 (o) = \left\{\left[g_{3}^\ast (o) \left(\frac{\partial \delta_{Y}(x;\beta^\ast)}{\partial \beta}\right)^{\rm T} + g_{4}^\ast (o) \left(\frac{\partial \mu_{Y}(t, z, x;\beta^\ast)}{\partial \beta}\right)^{\rm T}\right] (I\{x^{\rm T} \eta > 0\} - I\{x^{\rm T} \eta^\ast > 0\}) : \|\eta - \eta^\ast\|_2 < \delta \right\},     
\end{equation*}
\begin{equation*}
\mathcal{F}_\eta^4 (o) = \left\{g_{5}^\ast (o) \left(\frac{\partial \pi(t, z, x;\theta^\ast)}{\partial \theta}\right)^{\rm T} (I\{x^{\rm T} \eta > 0\} - I\{x^{\rm T} \eta^\ast > 0\}) : \|\eta - \eta^\ast\|_2 < \delta \right\},     
\end{equation*}
and let $B_2 = \sup |g_{1}^\ast (o) \partial \delta_{A}(x;\alpha^\ast) / \partial \alpha + g_{2}^\ast (o) \partial \mu_{A}(t, z, x;\alpha^\ast) / \partial \alpha|$, $B_3 = \sup |g_{3}^\ast (o) \partial \delta_{Y}(x;\beta^\ast) / \partial \beta + g_{4}^\ast (o) \partial \mu_{Y}(t, z, x;\beta^\ast) / \partial \beta|$, and $B_4 = \sup |g_{5}^\ast (o) \partial \pi(t, z, x;\theta^\ast) / \partial \theta|$, where $B_2, B_3, B_4 > 0$ and the supremum is taken over all the coordinates. By Assumption~\ref{asmp:posi} and Condition~\ref{cond.policy}, we have that $B_2, B_3, B_4 < \infty$.

Using the same technique as in {\bf Condition 2.1}, we define the envelop of $\mathcal{F}_\eta^j$ as $F_j = B_j I\{-k_0 \delta \leq x^{\rm T} \eta^\ast \leq k_0 \delta\}$ for $j = 2, 3, 4$, and obtain that
\begin{equation*}
\|F_j\|_{P,2} \leq \tilde{B}_j \delta^{1/2} < \infty, \quad j = 2, 3, 4,
\end{equation*}
where $\tilde{B}_2, \tilde{B}_3, \tilde{B}_4$ are some finite constants, and that $\mathcal{F}_\eta^j$ is a VC class with bounded bracketing entropy $J_{[]}^\ast(1,\mathcal{F}_\eta^j) < \infty$, for $j = 2, 3, 4$. By Theorem 11.2 of \cite{kosorok2008introduction}, we obtain that
\begin{equation*}
E^\ast \left[\sup_{\|\eta - \eta^\ast\|_2<\delta} \left\|\mathbb{G}_N \mathcal{F}_\eta^j\right\|_{1} \right] \leq c_j J_{[]}^\ast(1,\mathcal{F}_\eta^j) \|F_j\|_{P,2}, \quad j = 2, 3, 4,
\end{equation*}
where $c_2, c_3, c_4 > 0$ are some finite constants. 

Furthermore, by Theorem 2.14.5 of \cite{van1996weak}, we obtain that 
\begin{align*}
\left\{E^\ast \left[\sup_{\|\eta - \eta^\ast\|_2<\delta} \|\mathbb{G}_n \mathcal{F}_\eta^j\|_2^2\right]\right\}^{1/2}
& \leq l_j \left\{E^\ast \left[\sup_{\|\eta - \eta^\ast\|_2<\delta} \|\mathbb{G}_n \mathcal{F}_\eta^j\|_{1}\right] + \|F_j\|_{P,2}\right\} && \\
& \leq l_j \{c_j J_{[]}^\ast(1,\mathcal{F}_\eta^j) + 1\} \|F_j\|_{P,2} && \\
& \leq \tilde{c}_j \delta^{1/2}, && j = 2, 3, 4,
\end{align*}
where $l_2, l_3, l_4$ and $\tilde{c}_2, \tilde{c}_3, \tilde{c}_4$ are some finite constants.

By Equation~\eqref{eq:bound1}, we have that
\begin{align*}
(I) & = E^\ast \left[n^{1/2} \sup_{\|\eta - \eta^\ast\|_2 < \delta} \left|\hat{M}(\eta) - M_n^\ast(\eta) - \{\hat{M}(\eta^\ast) - M_n^\ast(\eta^\ast)\}\right| \right] \\
& \leq E^\ast \left[\sup_{\|\eta - \eta^\ast\|_2 < \delta} \left\{|\mathbb{G}_n \mathcal{F}_\eta^2(\hat{\alpha} - \alpha^\ast)| + |\mathbb{G}_n \mathcal{F}_\eta^3(\hat{\beta} - \beta^\ast)| + |\mathbb{G}_n \mathcal{F}_\eta^4(\hat{\theta} - \theta^\ast)| + o_p(1) \right\} \right] \\
& \leq n^{-1/2} \,\left\{E^\ast \left[\sup_{\|\eta - \eta^\ast\|_2 < \delta} |\mathbb{G}_n \mathcal{F}_\eta^2 \cdot n^{1/2} (\hat{\alpha} - \alpha^\ast)|\right] + E^\ast \left[\sup_{\|\eta - \eta^\ast\|_2 < \delta}|\mathbb{G}_n \mathcal{F}_\eta^3 \cdot n^{1/2}(\hat{\beta} - \beta^\ast)|\right] \right. \\
& \left. \qquad + E^\ast \left[\sup_{\|\eta - \eta^\ast\|_2 < \delta}|\mathbb{G}_n \mathcal{F}_\eta^4 \cdot n^{1/2} (\hat{\theta} - \theta^\ast)|\right] \right\},
\end{align*}
and then by the Cauchy-Schwarz inequality, we obtain that
\begin{align*}
(I) \leq & \, n^{-1/2} \left\{E[n \|\hat{\alpha} - \alpha^\ast\|_2^2]\right\}^{1/2} \left\{E^\ast \left[\sup_{\|\eta - \eta^\ast\|_2 < \delta} \|\mathbb{G}_n \mathcal{F}_\eta^2\|_2^2\right]\right\}^{1/2} \\
& + n^{-1/2} \left\{E[n \|\hat{\beta} - \beta^\ast\|_2^2]\right\}^{1/2} \left\{E^\ast \left[\sup_{\|\eta - \eta^\ast\|_2<\delta} \|\mathbb{G}_n \mathcal{F}_\eta^3\|_2^2\right]\right\}^{1/2} \\
& + n^{-1/2} \left\{E[n \|\hat{\theta} - \theta^\ast\|_2^2]\right\}^{1/2} \left\{E^\ast \left[\sup_{\|\eta - \eta^\ast\|_2<\delta} \|\mathbb{G}_n \mathcal{F}_\eta^4\|_2^2\right]\right\}^{1/2}.
\end{align*}

By Condition~\ref{cond.para}, we have that $B_\alpha = \left\{E[n \|\hat{\alpha} - \alpha^\ast\|_2^2]\right\}^{1/2} < \infty$, $B_\beta = \left\{E[n \|\hat{\beta} - \beta^\ast\|_2^2]\right\}^{1/2} < \infty$, $B_\theta = \left\{E[n \|\hat{\theta} - \theta^\ast\|_2^2]\right\}^{1/2} < \infty$, hence
\begin{equation*}
(I) \leq n^{-1/2} (B_\alpha \tilde{c}_2 + B_\beta \tilde{c}_3 + B_{\theta} \tilde{c}_4) \delta^{1/2}.
\end{equation*}

In summary, we conclude that as $n \to \infty$, the centered process satisfies
\begin{equation}\label{eq:cond2}
E^\ast \left[\sqrt{n} \sup_{\|\eta - \eta^\ast\|_2 < \delta} \left|\hat{M}(\eta) - M(\eta) - \{\hat{M}(\eta^\ast) - M(\eta^\ast)\}\right| \right] \leq (I) + (II) \leq \tilde{c}_1 \delta^{1/2}.
\end{equation}

Let $\phi_n(\delta) = \delta^{1/2}$ and $b = \frac{3}{2} < 2$, thus we have $\frac{\phi_n(\delta)}{\delta^b} = \delta^{-1}$ is decreasing, and $b$ does not depend on $n$.

{\bf Condition 3} By the facts that $\hat{\eta} \overset{p}{\to} \eta^\ast$ as $n \to \infty$, and that $\hat{M}(\hat{\eta}) \geq \sup_{\eta} \hat{M}(\eta)$, we choose $r_n = n^{1/3}$ such that $r_n^2 \phi_n(r_n^{-1}) = n^{2/3} \phi_n(n^{-1/3}) = n^{1/2}$.

In the end, the three conditions are satisfied with $r_n = n^{1/3}$; thus we conclude that $n^{1/3} \|\hat{\eta} - \eta^\ast\|_2 = O_p(1)$, which completes the proof of (i) of Theorem~\ref{thm:asym.para}.

{\bf PART 3.} We characterize the asymptotic distribution of $\hat{M}(\hat{\eta})$. First we note that 
\begin{equation*}
\sqrt{n} \{\hat{M}(\hat{\eta}) - M(\eta^\ast)\} = \sqrt{n} \{\hat{M}(\hat{\eta}) - \hat{M}(\eta^\ast)\} + \sqrt{n}\{\hat{M}(\eta^\ast) - M(\eta^\ast)\},
\end{equation*}
and then study the two terms in two steps.

{\bf Step 3.1} To establish $\sqrt{n} \{\hat{M}(\hat{\eta}) - \hat{M}(\eta^\ast)\} = o_p(1)$, it suffices to show that $\sqrt{n} \{M(\hat{\eta}) - M(\eta^\ast)\} = o_p(1)$ and $\sqrt{n} (\hat{M}(\hat{\eta}) - \hat{M}(\eta^\ast) - \{M(\hat{\eta}) - M(\eta^\ast)\}) = o_p(1)$.

First, as $n^{1/3} \|\hat{\eta} - \eta^\ast\|_2 = O_p(1)$, we apply the second-order Taylor expansion 
\begin{align*}
\sqrt{n} \{M(\hat{\eta}) - M(\eta^\ast)\} & = \sqrt{n} \left\{M'(\eta^\ast) \|\hat{\eta} - \eta^\ast\|_2 + \frac{1}{2} M''(\eta^\ast) \|\hat{\eta} - \eta^\ast\|_2^2 + o_p(\|\hat{\eta} - \eta^\ast\|_2^2) \right\} \\
& = \sqrt{n} \left\{\frac{1}{2} M''(\eta^\ast)\|\hat{\eta} - \eta^\ast\|_2^2 + o_p(\|\hat{\eta} - \eta^\ast\|_2^2) \right\} \\
& = \sqrt{n} \left\{\frac{1}{2} M''(\eta^\ast) O_p(n^{-2/3}) + o_p(n^{-2/3})\right\} = o_p(1),
\end{align*}
which proves (ii) of Theorem~\ref{thm:asym.para}.

Next, we follow the result~\eqref{eq:cond2} obtained in \textbf{PART 2}. As $n^{1/3} \|\hat{\eta} - \eta^\ast\|_2 = O_p(1)$, there exists $\tilde{\delta} = c_5 n^{-1/3}$, where $c_5 < \infty$ is a finite constant, such that $\|\hat{\eta} - \eta^\ast\|_2 \leq \tilde{\delta}$. Therefore we have
\begin{align*}
& \sqrt{n} (\hat{M}(\hat{\eta}) - \hat{M}(\eta^\ast) - \{M(\hat{\eta}) - M(\eta^\ast)\}) \\
& \leq E^\ast \left[\sqrt{n} \sup_{\|\hat{\eta} - \eta^\ast\|_2 < \tilde{\delta}} \left|\hat{M}(\hat{\eta}) - M(\hat{\eta}) - \{\hat{M}(\eta^\ast) - M(\eta^\ast)\}\right| \right] \\
& \leq \tilde{c}_1 \tilde{\delta}^{1/2} = \tilde{c}_1 \sqrt{c_5} n^{-1/6} = o_p(1),
\end{align*}
which yields the result.

{\bf Step 3.2} To derive the asymptotic distribution of $\sqrt{n}\{\hat{M}(\eta^\ast) - M(\eta^\ast)\}$, we follow the result~\eqref{eq:asym.dist} obtained in \textbf{PART 1} and have that
\begin{equation*}
\sqrt{n} \left\{\hat{M}(\eta^\ast) - M(\eta^\ast)\right\} \overset{D}{\to} \mathcal{N}(0, \sigma_{1}^2),
\end{equation*}
where $\sigma_{1}^2 = E[(M^\ast - M + H_{\alpha^\ast}^{\rm T} \phi_{\alpha, i} + H_{\beta^\ast}^{\rm T} \phi_{\beta, i} + H_{\theta^\ast}^{\rm T} \phi_{\theta, i})^2]$. 

Therefore we obtain in the end
\begin{align*}
\sqrt{n} \{\hat{M}(\hat{\eta}) - M(\eta^\ast)\} & = \sqrt{n} \{\hat{M}(\hat{\eta}) - \hat{M}(\eta^\ast)\} + \sqrt{n} \{\hat{M}(\eta^\ast) - M(\eta^\ast)\} \\
& = o_p(1) + \sqrt{n} \{\hat{M}(\eta^\ast) - M(\eta^\ast)\} \\
& \overset{D}{\to} \mathcal{N}(0, \sigma_{1}^2),
\end{align*}
which completes the proof.

\section{Proof of Theorem~\ref{thm:asym.np}}

We first review a useful lemma from \cite{kennedy2020sharp}, which illustrates the basic technique of cross-fitting.

\begin{lemma}\label{lem:cf}
Consider two independent samples $\mathcal{O}_1 = (O_1, \ldots, O_n)$ and $\mathcal{O}_2 = (O_{n+1}, \ldots, O_{\tilde{n}})$, let $\hat{f}(o)$ be a function estimated from $\mathcal{O}_2$ and $\mathbb{P}_n$ the empirical measure over $\mathcal{O}_1$, then we have
\begin{equation*}
    (\mathbb{P}_n - \mathbb{P})(\hat{f} - f) = O_{\mathbb{P}}\left(\frac{\|\hat{f} - f\|}{\sqrt{n}}\right)
\end{equation*}
\end{lemma}
\begin{proof}
First note that by conditioning on $\mathcal{O}_2$ we obtain
\begin{equation*}
    \mathbb{E} \left\{\mathbb{P}_n(\hat{f} - f) \,\big|\, \mathcal{O}_2\right\} = \mathbb{E} (\hat{f} - f \,|\, \mathcal{O}_2) = \mathbb{P}(\hat{f} - f)
\end{equation*}
and the conditional variance is
\begin{equation*}
    var\{(\mathbb{P}_n - \mathbb{P})(\hat{f} - f) \,|\, \mathcal{O}_2\} = var\{\mathbb{P}_n (\hat{f} - f) \,|\, \mathcal{O}_2\} = \frac{1}{n}var(\hat{f} - f \,|\, \mathcal{O}_2) \leq \|\hat{f} - f\|^2 / n
\end{equation*}
therefore by Chebyshev's inequality we have 
\begin{equation*}
    \mathbb{P}\left\{\frac{|(\mathbb{P}_n - \mathbb{P})(\hat{f} - f)|}{\|\hat{f} - f\|^2 / n} \geq t \right\} = \mathbb{E}\left[\mathbb{P}\left\{\frac{|(\mathbb{P}_n - \mathbb{P})(\hat{f} - f)|}{\|\hat{f} - f\|^2 / n} \geq t \,\bigg|\, \mathcal{O}_2\right\}\right] \leq \frac{1}{t^2}
\end{equation*}
thus for any $\epsilon > 0$ we can pick $t = 1 / \sqrt{\epsilon}$ so that the probability above is no more than $\epsilon$, which yields the result.
\end{proof}

We randomly split data into $K$ folds. For $k = 1, \ldots, K$, 
\begin{equation*}
\hat{M}(\eta) = \frac{1}{K} \sum_{k=1}^{K} \hat{M}_{k}(\eta) = \frac{1}{K} \sum_{k=1}^{K} P_{n,k} \{\Delta (O; \hat{\mu}_{A,-k}, \hat{\mu}_{Y,-k}, \hat{\pi}_{-k}) d(X)\},
\end{equation*}
where $P_{n,k}$ denote empirical averages only over the $k$-th fold, and $\hat{\mu}_{A,-k}$, $\hat{\mu}_{Y,-k}$ and $\hat{\pi}_{-k}$ denote the nuisance estimators constructed excluding the $k$-th fold.

Now we start our proof which has three main parts as follows.

{\bf PART 1.} We prove that $\hat{M}(\eta) - M_{n}(\eta) = o_p (n^{-1/2})$, where $M_{n}(\eta) = P_{n} \{\Delta(O) d(X,\eta)\}$. Essentially it suffices to prove that $\hat{M}_{k}(\eta) - M_{n,k}(\eta) = o_p (n^{-1/2})$, where $M_{n,k}(\eta) = P_{n,k} \{\Delta(O) d(X,\eta)\}$.

First we note the following decomposition
\begin{align*}
& \hat{M}_{k}(\eta) - M_{n,k}(\eta) \\
&= P_{n,k} \, d(\eta) \left\{ \frac{\hat{\delta}_{Y,-k}}{\hat{\delta}_{A,-k}} - \frac{\delta_{Y}}{\delta_{A}} + (2Z - 1) (2T - 1) \left[\left(\frac{1}{\hat{\pi}_{-k}} - \frac{1}{\pi}\right) \left(\frac{1}{\hat{\delta}_{A,-k}} - \frac{1}{\delta_{A}}\right) \left(Y - \hat{\mu}_{Y,-k} - \frac{\hat{\delta}_{Y,-k}}{\hat{\delta}_{A,-k}}(A - \hat{\mu}_{A,-k})\right) \right.\right.\\
&\left.\left. \qquad + \frac{1}{\delta_{A}} \left(\frac{1}{\hat{\pi}_{-k}} - \frac{1}{\pi}\right) G_1 + \frac{1}{\pi} \left(\frac{1}{\hat{\delta}_{A,-k}} - \frac{1}{\delta_{A}}\right) G_1 + \frac{1}{\pi \delta_{A}} G_2 \right.\right.\\
&\left.\left. \qquad + \frac{1}{\delta_{A}} \left(\frac{1}{\hat{\pi}_{-k}} - \frac{1}{\pi}\right) \left(Y - \mu_{Y} - \frac{\delta_{Y}}{\delta_{A}}(A - \mu_{A})\right) + \frac{1}{\pi} \left(\frac{1}{\hat{\delta}_{A,-k}} - \frac{1}{\delta_{A}}\right) \left(Y - \mu_{Y} - \frac{\delta_{Y}}{\delta_{A}}(A - \mu_{A})\right) \right.\right.\\
&\left.\left. \qquad + \frac{1}{\pi \delta_{A}} \left(\mu_{Y} - \hat{\mu}_{Y,-k} - \frac{1}{\delta_{A}}(\hat{\delta}_{Y,-k} - \delta_{Y})(A - \mu_{A}) - \delta_{Y} \left(\frac{1}{\hat{\delta}_{A,-k}} - \frac{1}{\delta_{A}}\right)(A - \mu_{A}) + \frac{\delta_{Y}}{\delta_{A}} (\hat{\mu}_{A,-k} - \mu_{A}) \right)\right]\right\},
\end{align*}
where we omit the arguments of the nuisance functions to simplify the notation, and denote
\begin{align*}
G_1 &= \mu_{Y} - \hat{\mu}_{Y,-k} - (\hat{\delta}_{Y,-k} - \delta_{Y}) \left(\frac{1}{\hat{\delta}_{A,-k}} - \frac{1}{\delta_{A}}\right) (A - \mu_{A}) - \frac{1}{\delta_{A}} (\hat{\delta}_{Y,-k} - \delta_{Y}) (A - \mu_{A}) \\
&\qquad + \frac{1}{\delta_{A}} (\hat{\delta}_{Y,-k} - \delta_{Y}) (\hat{\mu}_{A,-k} - \mu_{A}) - \delta_{Y} \left(\frac{1}{\hat{\delta}_{A,-k}} - \frac{1}{\delta_{A}}\right) (A - \mu_{A}) \\
&\qquad + \delta_{Y} \left(\frac{1}{\hat{\delta}_{A,-k}} - \frac{1}{\delta_{A}}\right) (\hat{\mu}_{A,-k} - \mu_{A}) + \frac{\delta_{Y}}{\delta_{A}} (\hat{\mu}_{A,-k} - \mu_{A}),
\end{align*}
\begin{equation*}
G_2 = \frac{\hat{\delta}_{Y,-k} - \delta_{Y}}{\delta_{A}} (\hat{\mu}_{A,-k} - \mu_{A}) + \delta_{Y} \left(\frac{1}{\hat{\delta}_{A,-k}} - \frac{1}{\delta_{A}}\right) (\hat{\mu}_{A,-k} - \mu_{A}) - (\hat{\delta}_{Y,-k} - \delta_{Y}) \left(\frac{1}{\hat{\delta}_{A,-k}} - \frac{1}{\delta_{A}}\right) (A - \mu_{A}).
\end{equation*}

In summary, we have two types of terms from this decomposition: product terms and mean zero terms (by multiple robustness). The product terms are $o_p (n^{-1/2})$ by Cauchy-Schwarz inequality and Condition~\ref{cond.np} (rate of convergence). The mean zero terms are $o_p (n^{-1/2})$ by Lemma~\ref{lem:cf}.

{\bf PART 2.} We prove that $n^{1/3} \|\hat{\eta} - \eta^\ast\|_2 = O_p(1)$.

First we note that, by Condition~\ref{cond.policy} (iii), $M(\eta)$ is twice continuously differentiable at a neighborhood of $\eta^\ast$. In PART 1, we show that $\hat{M}(\eta) = M(\eta) + o_p(1), \forall \eta$. Since $\hat{\eta}$ maximizes $\hat{M}(\eta)$, we have that $\hat{M}(\hat{\eta}) \geq \sup_{\eta} \hat{M}(\eta)$; thus by the Argmax theorem, we obtain that $\hat{\eta} \overset{p}{\to} \eta^\ast$ as $n \to \infty$. 

Then we apply Theorem 14.4 (Rate of convergence) of \cite{kosorok2008introduction} to establish the $n^{-1/3}$ rate of convergence of $\hat{\eta}$, and need to find the suitable rate that satisfies three conditions below.

{\bf Condition 1} For every $\eta$ in a neighborhood of $\eta^\ast$ such that $\|\eta - \eta^\ast\|_2 < \delta$, by Condition~\ref{cond.policy} (iii), we apply the second-order Taylor expansion,
\begin{align*}
M(\eta) - M(\eta^\ast) &= M'(\eta^\ast) \|\eta - \eta^\ast\|_2 + \frac{1}{2} M''(\eta^\ast) \|\eta - \eta^\ast\|_2^2 + o(\|\eta - \eta^\ast\|_2^2) \\
&= \frac{1}{2} S''(\eta^\ast)\|\eta - \eta^\ast\|_2^2 + o(\|\eta - \eta^\ast\|_2^2),
\end{align*}
and as $S''(\eta^\ast) < 0$, there exists $c_0 = -\frac{1}{2} S''(\eta^\ast) > 0$ such that $S(t;\eta) - S(t;\eta^\ast) \leq c_0 \|\eta - \eta^\ast\|_2^2$. 

{\bf Condition 2} For all $n$ large enough and sufficiently small $\delta$, we consider the centered process $\hat{M} - M$, and have that 
\begin{align*}
& E^\ast \left[ \sqrt{n} \sup_{\|\eta - \eta^\ast\|_2 < \delta} \left| \hat{M}(\eta) - M(\eta) - \left\{\hat{M}(\eta^\ast) - M(\eta^\ast)\right\} \right| \right] \\
&= E^\ast \left[ \sqrt{n} \sup_{\|\eta - \eta^\ast\|_2 < \delta} \left| \hat{M}(\eta) - M_n^\ast(\eta) + M_n^\ast(\eta) - M(\eta) - \left\{ \hat{M}(\eta^\ast) - M_n^\ast(\eta^\ast) + M_n^\ast(\eta^\ast) - M(\eta^\ast)\right\} \right| \right] \\
& \leq E^\ast \left[ \sqrt{n} \sup_{\|\eta - \eta^\ast\|_2 < \delta} \left| \hat{M}(\eta) - M_n^\ast(\eta) - \left\{ \hat{M}(\eta^\ast) - M_n^\ast(\eta^\ast) \right\} \right| \right] \\
& \quad + E^\ast \left[ \sqrt{n} \sup_{\|\eta - \eta^\ast\|_2 < \delta} \left| M_n^\ast(\eta) - M(\eta) - \left\{ M_n^\ast(\eta^\ast) - M(\eta^\ast)\right\} \right| \right] \\
&= (I) + (II),
\end{align*}
where $E^\ast (\cdot)$ denote the outer expectation, and we bound $(I)$ and $(II)$ respectively as follows.

It follows from the result in \textbf{PART 1} that $(I) = o_p(1)$.

To bound $(II)$, we note that
\begin{align*}
M_n^\ast (\eta) - M_n^\ast (\eta^\ast) &= \frac{1}{n} \sum_{i=1}^{n} \Delta^\ast (O_i) (d(X_i;\eta) - d(X_i;\eta\ast)) \\
&= \frac{1}{n} \sum_{i=1}^{n} \Delta^\ast (O_i) (I\{X_{i}^{\rm T} \eta > 0\} - I\{X_{i}^{\rm T} \eta^\ast > 0\}),
\end{align*}
where
\begin{equation*}
\Delta^\ast (o) = \frac{\delta_{Y}(x)}{\delta_{A}(x)} + \frac{(2z - 1)(2t - 1)}{\pi(t, z, x) \delta_{A}(x)} \left\{y - \mu_{Y}(t, z, x) - \frac{\delta_{Y}(x)}{\delta_{A}(x)} (a - \mu_{A}(t, z, x)) \right\}.
\end{equation*}

We define a class of functions
\begin{equation*}
\mathcal{F}_\eta^5 (o) = \left\{\Delta^\ast (o) (I\{x^{\rm T} \eta > 0\} - I\{x^{\rm T} \eta^\ast > 0\}) : \|\eta - \eta^\ast\|_2 < \delta \right\},     
\end{equation*}
and let $B_5 = \sup |\Delta^\ast (o)|$. By Assumption~\ref{asmp:posi} and Condition~\ref{cond.policy}, we have that $B_5 < \infty$. 

Using the same technique as in Section {\bf Condition 2.1}, we define the envelop of $\mathcal{F}_\eta^5$ as $F_5 = B_5 I\{-k_0 \delta \leq x^T \eta^\ast \leq k_0 \delta\}$, and obtain that $\|F_5\|_{P,2} \leq \tilde{B}_9 \delta^{1/2} < \infty$, where $\tilde{B}_9$ is a finite constant, and that $\mathcal{F}_\eta^5$ is a VC class with bounded entropy $J_{[]}^\ast(1,\mathcal{F}_\eta^5) < \infty$. By Theorem 11.2 of \cite{kosorok2008introduction}, we obtain that there exists a constant $0 < c_6 < \infty$,
\begin{equation*}
(II) = E^\ast \left[\sup_{\|\eta - \eta^\ast\|_2 < \delta} |\mathbb{G}_n \mathcal{F}_\eta^5|\right] \leq c_6 J_{[]}^\ast(1,\mathcal{F}_\eta^5) \|F_5\|_{P,2} \leq c_6 J_{[]}^\ast(1,\mathcal{F}_\eta^5)B_5\sqrt{2k_0 k_1} \delta^{1/2} = \tilde{c}_5 \delta^{1/2}.
\end{equation*}

In summary, we conclude that as $n \to \infty$, the centered process satisfies
\begin{equation}\label{eq:cond2.np}
E^\ast \left[\sqrt{n} \sup_{\|\eta - \eta^\ast\|_2 < \delta} \left|\hat{M}(\eta) - M(\eta) - \{\hat{M}(\eta^\ast) - M(\eta^\ast)\}\right| \right] \leq (I) + (II) \leq \tilde{c}_5 \delta^{1/2}.
\end{equation}

Let $\phi_n(\delta) = \delta^{1/2}$ and $b = \frac{3}{2} < 2$, thus we have $\frac{\phi_n(\delta)}{\delta^b} = \delta^{-1}$ is decreasing, and $b$ does not depend on $n$.

{\bf Condition 3} By the facts that $\hat{\eta} \overset{p}{\to} \eta^\ast$ as $n \to \infty$, and that $\hat{M}(\hat{\eta}) \geq \sup_{\eta} \hat{M}(\eta)$, we choose $r_n = n^{1/3}$ such that $r_n^2 \phi_n(r_n^{-1}) = n^{2/3} \phi_n(n^{-1/3}) = n^{1/2}$.

In the end, the three conditions are satisfied with $r_n = n^{1/3}$; thus we conclude that $n^{1/3} \|\hat{\eta} - \eta^\ast\|_2 = O_p(1)$, which completes the proof of (i) of Theorem~\ref{thm:asym.np}.

{\bf PART 3.} We characterize the asymptotic distribution of $\hat{M}(\hat{\eta})$. First we note that 
\begin{equation*}
\sqrt{n} \{\hat{M}(\hat{\eta}) - M(\eta^\ast)\} = \sqrt{n} \{\hat{M}(\hat{\eta}) - \hat{M}(\eta^\ast)\} + \sqrt{n}\{\hat{M}(\eta^\ast) - M(\eta^\ast)\},
\end{equation*}
and then study the two terms in two steps.

{\bf Step 3.1} To establish $\sqrt{n} \{\hat{M}(\hat{\eta}) - \hat{M}(\eta^\ast)\} = o_p(1)$, it suffices to show that $\sqrt{n} \{M(\hat{\eta}) - M(\eta^\ast)\} = o_p(1)$ and $\sqrt{n} (\hat{M}(\hat{\eta}) - \hat{M}(\eta^\ast) - \{M(\hat{\eta}) - M(\eta^\ast)\}) = o_p(1)$.

First, as $n^{1/3} \|\hat{\eta} - \eta^\ast\|_2 = O_p(1)$, we apply the second-order Taylor expansion 
\begin{align*}
\sqrt{n} \{M(\hat{\eta}) - M(\eta^\ast)\} & = \sqrt{n} \left\{M'(\eta^\ast) \|\hat{\eta} - \eta^\ast\|_2 + \frac{1}{2} M''(\eta^\ast) \|\hat{\eta} - \eta^\ast\|_2^2 + o_p(\|\hat{\eta} - \eta^\ast\|_2^2) \right\} \\
& = \sqrt{n} \left\{\frac{1}{2} M''(\eta^\ast)\|\hat{\eta} - \eta^\ast\|_2^2 + o_p(\|\hat{\eta} - \eta^\ast\|_2^2) \right\} \\
& = \sqrt{n} \left\{\frac{1}{2} M''(\eta^\ast) O_p(n^{-2/3}) + o_p(n^{-2/3})\right\} = o_p(1),
\end{align*}
which proves (ii) of Theorem~\ref{thm:asym.np}.

Next, we follow the result~\eqref{eq:cond2.np} obtained in \textbf{PART 2}. As $n^{1/3} \|\hat{\eta} - \eta^\ast\|_2 = O_p(1)$, there exists $\tilde{\delta} = c_7 n^{-1/3}$, where $c_7 < \infty$ is a finite constant, such that $\|\hat{\eta} - \eta^\ast\|_2 \leq \tilde{\delta}$. Therefore we have
\begin{align*}
& \sqrt{n} (\hat{M}(\hat{\eta}) - \hat{M}(\eta^\ast) - \{M(\hat{\eta}) - M(\eta^\ast)\}) \\
& \leq E^\ast \left[\sqrt{n} \sup_{\|\hat{\eta} - \eta^\ast\|_2 < \tilde{\delta}} \left|\hat{M}(\hat{\eta}) - M(\hat{\eta}) - \{\hat{M}(\eta^\ast) - M(\eta^\ast)\}\right| \right] \\
& \leq \tilde{c}_5 \tilde{\delta}^{1/2} = \tilde{c}_5 \sqrt{c_7} n^{-1/6} = o_p(1),
\end{align*}
which yields the result.

{\bf Step 3.2} To derive the asymptotic distribution of $\sqrt{n}\{\hat{M}(\eta^\ast) - M(\eta^\ast)\}$, we follow the result obtained in \textbf{PART 1} and have that
\begin{equation*}
\sqrt{n} \left\{\hat{M}(\eta^\ast) - M(\eta^\ast)\right\} \overset{D}{\to} \mathcal{N}(0, \sigma_{2}^2),
\end{equation*}
where $\sigma_{2}^2 = E[(\Delta(O_{i}) d(X_{i};\eta^\ast) - M(\eta^\ast))^2]$. 

Therefore we obtain in the end
\begin{align*}
\sqrt{n} \{\hat{M}(\hat{\eta}) - M(\eta^\ast)\} & = \sqrt{n} \{\hat{M}(\hat{\eta}) - \hat{M}(\eta^\ast)\} + \sqrt{n} \{\hat{M}(\eta^\ast) - M(\eta^\ast)\} \\
& = o_p(1) + \sqrt{n} \{\hat{M}(\eta^\ast) - M(\eta^\ast)\} \\
& \overset{D}{\to} \mathcal{N}(0, \sigma_{2}^2),
\end{align*}
which completes the proof.

\section{Proof of Theorem~\ref{thm:panel.cate} and \ref{thm:panel.policy}}
\label{sm.sec:panel.idid}

We first prove the identification result.

First we note that
\begin{align*}
& \delta_{Y,1}(X) - \delta_{Y,0}(X) = E[Y_{1} - Y_{0} \mid X, Z = 1] - E[Y_{1} - Y_{0} \mid X, Z = 0] \\
&= \sum_{z=0,1} (2z - 1) E[Y_{1} - Y_{0} \mid X, Z = z] \\
&= \sum_{z=0,1} (2z - 1) E[Y_{1}(A_{1}(z)) - Y_{0}(A_{0}(z)) \mid X, Z = z] \\
&= \sum_{z=0,1} (2z - 1) E[A_{1}(z) Y_{1}(1) + (1 - A_{1}(z)) Y_{1}(0) - A_{0}(z) Y_{0}(1) - (1 - A_{0}(z)) Y_{0}(0) \mid Z = z, X] \\
&= \sum_{z=0,1} (2z - 1) E[A_{1}(z) (Y_{1}(1) - Y_{1}(0)) - A_{0}(z) (Y_{0}(1) - Y_{0}(0)) + Y_{1}(0) - Y_{0}(0) \mid Z = z, X] \\
&= \sum_{z=0,1} (2z - 1) (E[A_{1}(z) (Y_{1}(1) - Y_{1}(0)) \mid X, Z = z] - E[A_{0}(z) (Y_{0}(1) - Y_{0}(0)) \mid X, Z = z] \\
&\qquad\qquad\qquad\qquad + E[Y_{1}(0) - Y_{0}(0) \mid X, Z = z]) \\
&= \sum_{z=0,1} (2z - 1) (E[A_{1}(z) (Y_{1}(1) - Y_{1}(0)) \mid X, Z = z] - E[A_{0}(z) (Y_{0}(1) - Y_{0}(0)) \mid X, Z = z] \\
&= \sum_{z=0,1} (2z - 1) (E[A_{1}(z) (Y_{1}(1) - Y_{1}(0)) \mid X] - E[A_{0}(z) (Y_{0}(1) - Y_{0}(0)) \mid X] \\
&= E[(A_{1}(1) - A_{1}(0)) (Y_{1}(1) - Y_{1}(0)) \mid X] - E[(A_{0}(1) - A_{0}(0)) (Y_{0}(1) - Y_{0}(0)) \mid X] \\
&= E[A_{1}(1) - A_{1}(0) \mid X] \tau(X) - E[A_{0}(1) - A_{0}(0) \mid X] \tau(X) \\
&= E[A_{1}(1) - A_{1}(0) - A_{0}(1) + A_{0}(0) \mid X] \tau(X).
\end{align*}

We also note that 
\begin{align*}
& \delta_{A,1}(X) - \delta_{A,0}(X) = E[A_{1} - A_{0} \mid X, Z = 1] - E[A_{1} - A_{0} \mid X, Z = 0] \\
&= \sum_{z=0,1} (2z - 1) E[A_{1} - A_{0} \mid X, Z = z] \\
&= \sum_{z=0,1} (2z - 1) E[A_{1}(z) - A_{0}(z) \mid X, Z = z] \\
&= E[A_{1}(1) - A_{1}(0) - A_{0}(1) + A_{0}(0) \mid X].
\end{align*}

Combining the above derivations, we obtain that $\delta_{Y,1}(X) - \delta_{Y,0}(X) = (\delta_{A,1}(X) - \delta_{A,0}(X)) \tau(X)$. That is, the CATE is identified by
\begin{equation*}
\tau(X) = \frac{\delta_{Y,1}(X) - \delta_{Y,0}(X)}{\delta_{A,1}(X) - \delta_{A,0}(X)}.
\end{equation*}

Alternatively, we consider the following assumptions: (sequential ignorability) $Y_{t}(a) \perp A_{t} \mid U, X, Z$ for $t, a = 0, 1$, and there is no additive interaction of either {\rm (i)} $E[A_{1} - A_{0} \mid X, U, Z = 1] - E[A_{1} - A_{0} \mid X, U, Z = 0] = E[A_{1} - A_{0} \mid X, Z = 1] - E[A_{1} - A_{0} \mid X, Z = 0]$ or {\rm (ii)} $E[Y_{t}(1) - Y_{t}(0) \mid U, X] = E[Y_{t}(1) - Y_{t}(0) \mid X]$ for $t = 0, 1$.

We can continue that 
\begin{align*}
& \delta_{Y,1}(X) - \delta_{Y,0}(X) \\
&= \sum_{z=0,1} (2z - 1) (E[A_{1}(z) (Y_{1}(1) - Y_{1}(0)) \mid X, Z = z] - E[A_{0}(z) (Y_{0}(1) - Y_{0}(0)) \mid X, Z = z]) \\
&= E_{U} \sum_{z=0,1} (2z - 1) (E[A_{1}(z) (Y_{1}(1) - Y_{1}(0)) \mid X, U, Z = z] - E[A_{0}(z) (Y_{0}(1) - Y_{0}(0)) \mid X, U, Z = z] \\
&= E_{U} \sum_{z=0,1} (2z - 1) (E[A_{1}(z) \mid X, U, Z = z] E[Y_{1}(1) - Y_{1}(0) \mid X, U, Z = z] \\
&\qquad\qquad\qquad\qquad\quad - E[A_{0}(z) \mid X, U, Z = z] E[Y_{0}(1) - Y_{0}(0) \mid X, U, Z = z]) \\
&= E_{U} [E[Y_{t}(1) - Y_{t}(0) \mid U, X] (E[A_{1} - A_{0} \mid X, U, Z = 1] - E[A_{1} - A_{0} \mid X, U, Z = 0])].
\end{align*}

Under Assumption (i), we have that
\begin{align*}
& E[A_{1} - A_{0} \mid X, U, Z = 1] - E[A_{1} - A_{0} \mid X, U, Z = 0] \\
&= E[A_{1} - A_{0} \mid X, Z = 1] - E[A_{1} - A_{0} \mid X, Z = 0] \\
&= \delta_{A,1}(X) - \delta_{A,0}(X);
\end{align*}
or under Assumption (ii), we have that
\begin{equation*}
E[Y_{t}(1) - Y_{t}(0) \mid U, X] = E[Y_{t}(1) - Y_{t}(0) \mid X], t = 0, 1,
\end{equation*}
and also
\begin{equation*}
E_{U} [E[A_{1} - A_{0} \mid X, U, Z = 1] - E[A_{1} - A_{0} \mid X, U, Z = 0]] = \delta_{A,1}(X) - \delta_{A,0}(X).
\end{equation*}

Hence combining the above derivations, we obtain the same identification results.

Next, we derive the efficient influence function.

For a given distribution $P$ in the nonparametric statistical model $\mathcal{M}$, we let $p$ denote the density of $P$ with respect to some dominating measure $\nu$. For all bounded $h \in L_{2}(P)$, define the parametric submodel $p_{\epsilon} = (1 + \epsilon h) p$, which is valid for small enough $\epsilon$ and has score $h$ at $\epsilon = 0$.

We study the following statistical functional
\begin{equation*}
\Psi (P) = E_{P} \left[\frac{E_{P}[Y_{1} - Y_{0} \mid X = x, Z = 1] - E_{P}[Y_{1} - Y_{0} \mid X = x, Z = 0]}{E_{P}[A_{1} - A_{0} \mid X = x, Z = 1] - E_{P}[A_{1} - A_{0} \mid X = x, Z = 0]}\right],
\end{equation*}
and would establish that $\Psi (P)$ is pathwise differentiable with respect to $\mathcal{M}$ at $P$ with efficient influence function $\phi_{P}$ if we have that for any $P \in \mathcal{M}$
\begin{equation*}
\frac{\partial}{\partial\epsilon} \Psi (P_{\epsilon}) \bigg|_{\epsilon = 0} = \int \phi_{P} (o) h(o) dP(o).
\end{equation*}

We denote $\delta_{Y,z,\epsilon}(x) = E_{P_{\epsilon}}[Y_1 - Y_0 \mid X = x, Z = z]$, $\delta_{A,z,\epsilon}(x) = E_{P_{\epsilon}}[A_1 - A_0 \mid X = x, Z = z]$, $S = \partial \log p_{\epsilon} / \partial\epsilon$, and compute
\begin{align*}
\frac{\partial}{\partial\epsilon} \Psi (P_{\epsilon}) \bigg|_{\epsilon = 0} &= \frac{\partial}{\partial\epsilon} E_{P_{\epsilon}}\left[\frac{\delta_{Y,1,\epsilon}(X) - \delta_{Y,0,\epsilon}(X)}{\delta_{A,1,\epsilon}(X) - \delta_{A,0,\epsilon}(X)}\right] \bigg|_{\epsilon = 0} \\
&= \frac{\partial}{\partial\epsilon} E_{P}\left[(1 + \epsilon S) \frac{\delta_{Y,1,\epsilon}(X) - \delta_{Y,0,\epsilon}(X)}{\delta_{A,1,\epsilon}(X) - \delta_{A,0,\epsilon}(X)}\right] \bigg|_{\epsilon = 0} \\
&= E_{P}\left[S \frac{\delta_{Y,1}(X) - \delta_{Y,0}(X)}{\delta_{A,1}(X) - \delta_{A,0}(X)}\right] \\
&\quad + E_{P}\left[\frac{1}{\delta_{A,1}(X) - \delta_{A,0}(X)} \left(\frac{\partial}{\partial\epsilon} \delta_{Y,1,\epsilon}(X) \bigg|_{\epsilon = 0} - \frac{\partial}{\partial\epsilon} \delta_{Y,0,\epsilon}(X) \bigg|_{\epsilon = 0}\right)\right] \\
&\quad - E_{P}\left[\frac{\delta_{Y,1}(X) - \delta_{Y,0}(X)}{\{\delta_{A,1}(X) - \delta_{A,0}(X)\}^2}\left(\frac{\partial}{\partial\epsilon} \delta_{A,1,\epsilon}(X) \bigg|_{\epsilon = 0} - \frac{\partial}{\partial\epsilon} \delta_{A,0,\epsilon}(X) \bigg|_{\epsilon = 0}\right)\right].
\end{align*}

Then we need to compute
\begin{align*}
\frac{\partial}{\partial\epsilon} \delta_{Y,z,\epsilon}(X) \bigg|_{\epsilon = 0} &= \frac{\partial}{\partial\epsilon} E_{P_{\epsilon}}[Y_1 - Y_0 \mid X, Z = z] \bigg|_{\epsilon = 0} \\
&= \frac{\partial}{\partial\epsilon} \frac{\delta_{Y,z}(X) + \epsilon E_{P}[S (Y_1 - Y_0) \mid X, Z = z]}{1 + \epsilon E_{P}[S \mid X, Z = z]} \bigg|_{\epsilon = 0} \\
&= E_{P}[S (Y_1 - Y_0) \mid X, Z = z] - \delta_{Y,z}(X) E_{P}[S \mid X, Z = z] \\
&= E_{P}\left[S \frac{(Y_1 - Y_0 - \delta_{Y,z}(X)) I\{Z = z\}}{z \pi_{Z} (X) + (1-z)(1 - \pi_{Z} (X))} \mid X \right],
\end{align*}
and 
\begin{align*}
\frac{\partial}{\partial\epsilon} \delta_{A,z,\epsilon}(X) \bigg|_{\epsilon = 0} &= \frac{\partial}{\partial\epsilon} E_{P_{\epsilon}}[A_1 - A_0 \mid X, Z = z] \bigg|_{\epsilon = 0} \\
&= \frac{\partial}{\partial\epsilon} \frac{\delta_{A,z}(X) + \epsilon E_{P}[S (A_1 - A_0) \mid X, Z = z]}{1 + \epsilon E_{P}[S \mid X, Z = z]} \bigg|_{\epsilon = 0} \\
&= E_{P}[S (A_1 - A_0) \mid X, Z = z] - \delta_{A,z}(X) E_{P}[S \mid X, Z = z] \\
&= E_{P}\left[S \frac{(A_1 - A_0 - \delta_{A,z}(X)) I\{Z = z\}}{z \pi_{Z} (X) + (1-z)(1 - \pi_{Z} (X))} \mid X \right].
\end{align*}

In summary, we obtain the efficient influence function
\begin{align*}
\phi_{P}(O) &= \frac{E[Y_{1} - Y_{0} \mid X, Z = 1] - E[Y_{1} - Y_{0} \mid X, Z = 0]}{E[A_{1} - A_{0} \mid X, Z = 1] - E[A_{1} - A_{0} \mid X, Z = 0]} \\
&\quad + \frac{Z - \pi_{Z}(X)}{\pi_{Z}(X) (1 - \pi_{Z}(X)) (\delta_{A,1}(X) - \delta_{A,0}(X))^2} \left\{(Y_{1} - Y_{0})(\delta_{A,1}(X) - \delta_{A,0}(X)) \right. \\
&\quad \left. - (A_{1} - A_{0})(\delta_{Y,1}(X) - \delta_{Y,0}(X)) + \delta_{Y,1}(X) \delta_{A,0}(X) - \delta_{Y,0}(X) \delta_{A,1}(X)\right\} - \Psi(P).
\end{align*}

Finally, it follows to prove Theorem~\ref{thm:panel.policy} by Equation~\eqref{eq:opt.policy}.

\section{Additional simulations}
\label{sm.sec:add.simu}

In this section, we report additional simulation results to illustrate how different sample sizes and the strength of the IV affect the performance of the estimated policies.

\subsection{Sensitivity analysis}

In this section, we study how the strength of the IV affects the performance of the estimated policies. The data generation process is the same as Section~\ref{sec:simu}, except that the treatment assignment mechanism is given by
\begin{align*}
Pr(A_{0} = 1 \mid Z, U, X) & = \text{expit}(1.5 - 3 Z + 0.2 U_{0} + 2 X_{1}), \\
Pr(A_{1} = 1 \mid Z, U, X) & = \text{expit}(-1.5 + 2 Z - 0.15 U_{1} + 1.5 X_{2}),
\end{align*}
for weak IV strength, and 
\begin{align*}
Pr(A_{0} = 1 \mid Z, U, X) & = \text{expit}(3 - 7 Z + 0.2 U_{0} + 2 X_{1}), \\
Pr(A_{1} = 1 \mid Z, U, X) & = \text{expit}(-3 + 5 Z - 0.15 U_{1} + 1.5 X_{2}),
\end{align*}
for strong IV strength, respectively. Simulation results are reported in Figures~\ref{fig:str.para} and \ref{fig:str.ml}. 

\begin{figure}[p]
    \centering
    \includegraphics[scale=0.8]{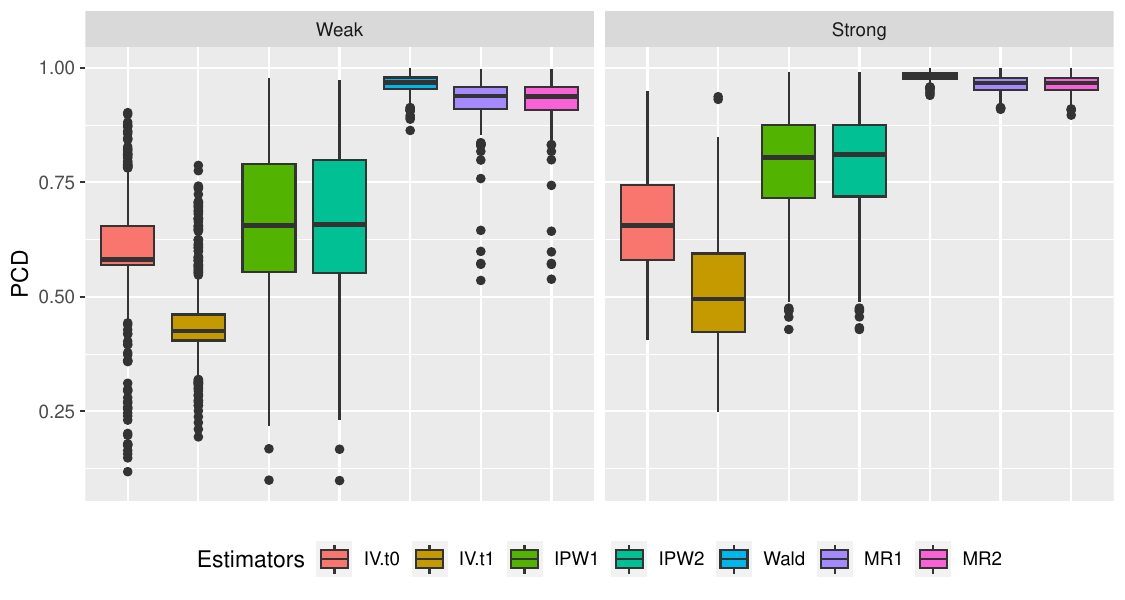}
    \caption{The percentage of correct decisions (PCD) results of the estimated optimal policies using parametric models, under weak (left) or strong (right) IV strength.}
    \label{fig:str.para}
\end{figure}

\begin{figure}[p]
    \centering
    \includegraphics[scale=0.8]{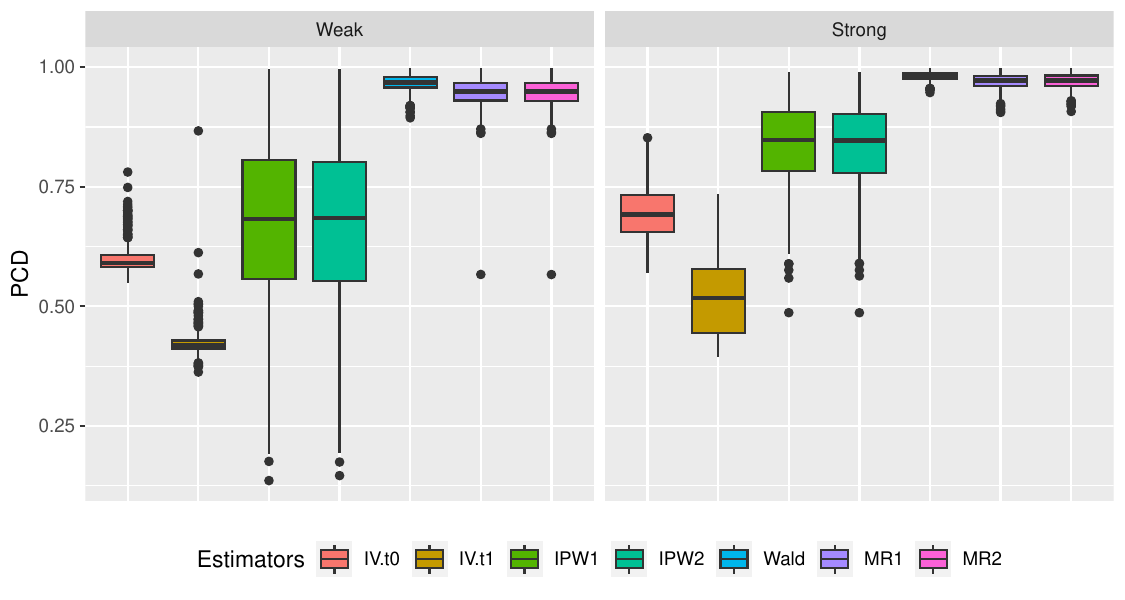}
    \caption{The percentage of correct decisions (PCD) results of the estimated optimal policies using machine learning, under weak (left) or strong (right) IV strength.}
    \label{fig:str.ml}
\end{figure}

\subsection{Sample size}

In this section, we study how different sample sizes affect the performance of the estimated policies. The data generation process is the same as Section~\ref{sec:simu}. The sample sizes are $n = 2500, 10000$ when using parametric models, and $n = 5000, 20000$ when using machine learning. Simulation results are reported in Figures~\ref{fig:ss.para} and \ref{fig:ss.ml}. 

\begin{figure}[p]
    \centering
    \includegraphics[scale=0.8]{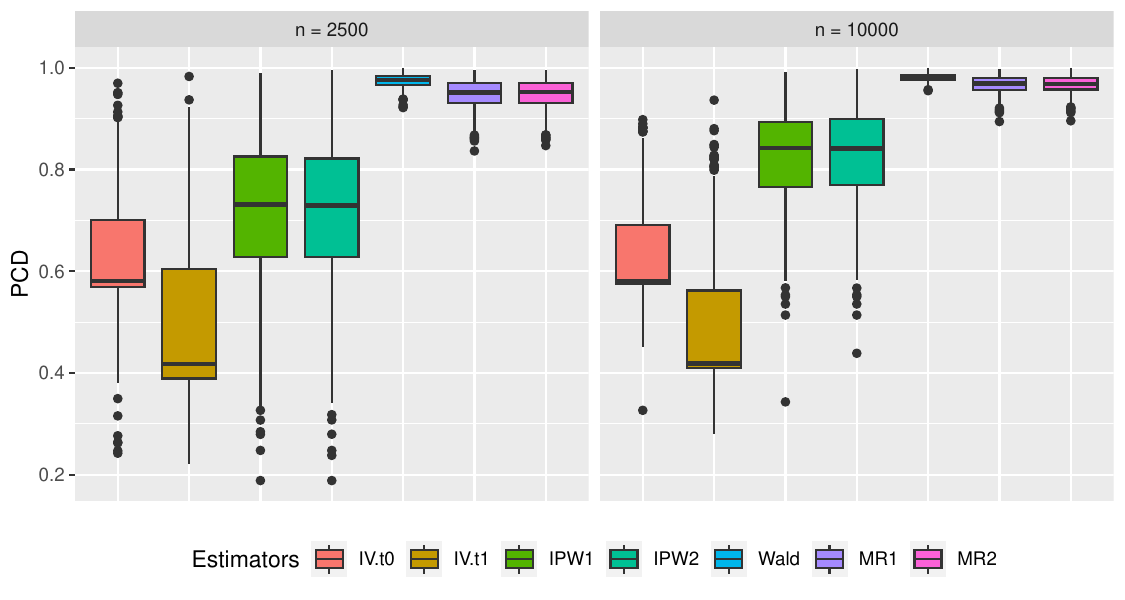}
    \caption{The percentage of correct decisions (PCD) results of the estimated optimal policies, using parametric models with sample size $n=2500$ (left) or $n=10000$ (right).}
    \label{fig:ss.para}
\end{figure}

\begin{figure}[p]
    \centering
    \includegraphics[scale=0.8]{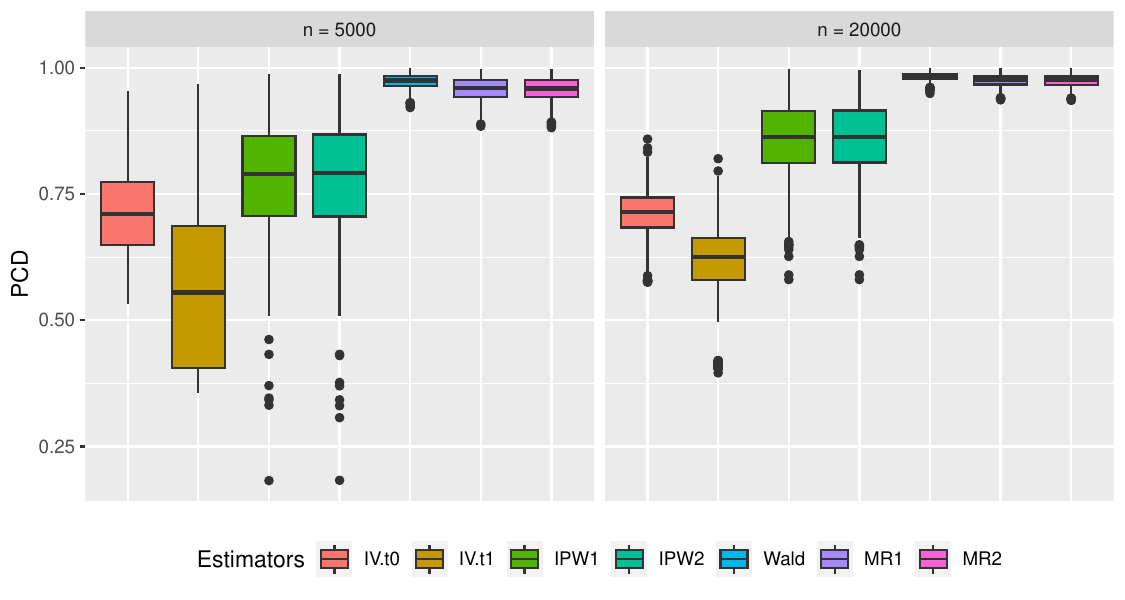}
    \caption{The percentage of correct decisions (PCD) results of the estimated optimal policies, using machine learning with sample size $n=5000$ (left) or $n=20000$ (right).}
    \label{fig:ss.ml}
\end{figure}

\section{Australian Longitudinal Survey}
\label{sm.sec:als.data}

In this section, we provide supplementary information on our data analysis of the Australian Longitudinal Survey. The data can be accessed by making a request to the \href{https://dataverse.ada.edu.au/dataverse/australian-longitudinal-survey}{Australian Data Archive} (Australian National University).

We follow \cite{su2013local,cai2006functional} and use an index of labor market attitudes as the instrumental variable in our analysis. The survey includes seven questions about work, social roles and school attitudes towards working women. Individuals respond to these questions with scores ($1$) strongly agree, ($2$) agree, ($3$) don’t know, ($4$) disagree, and ($5$) strongly disagree. This survey design implies that a response with a higher score indicates more positive attitude towards the education benefit of women and also their active role in the labor market. Following \cite{su2013local}, we use only six out of the seven questions to construct our attitudes index, since questions 2 and 3 are actually very similar, thus might be repetitive. We choose question 2 over question 3.
Summary statistics of our data from the 1984 and 1985 waves are reported in Table~\ref{tab:ALS_sample1984} and \ref{tab:ALS_sample1985}, respectively. Replication code is available at \href{https://github.com/panzhaooo/policy-learning-instrumented-DiD}{GitHub}.

\begin{table}[p]
    \centering
    \begin{tabular}{c c c c c c}
        \hline\hline
        Variable & Source & Mean & SD & Min & Max \\
        \hline
        \texttt{born\_australia} & A12                             & $0.82$ & $0.38$ & $0$ & $1$ \\
        \texttt{married}         & A9                              & $0.07$ & $0.25$ & $0$ & $1$ \\
        \texttt{uni\_mem}        & G10                             & $0.34$ & $0.48$ & $0$ & $1$ \\
        \texttt{gov\_emp}        & G9                              & $0.21$ & $0.41$ & $0$ & $1$ \\
        \texttt{age}             & A4                              & $20.07$ & $2.45$ & $14$ & $26$ \\
        \texttt{year\_expe}      & F3-4, F7-10, F31-33, G21-23     & $0.94$ & $1.40$ & $0$ & $11$ \\
        \texttt{attitude}        & O1-7                            & $17.94$ & $3.48$ & $6$ & $28$ \\
        \texttt{year\_edu}       & E4, E7, E10, E14, E16, E23, E25 & $11.14$ & $1.93$ & $3$ & $20$ \\
        \texttt{wage\_hour}      & G3-5, G7-8                      & $4.83$ & $2.01$ & $0.57$ & $21.43$ \\
        \hline\hline
    \end{tabular}
    \caption{The 1984 wave summary statistics of variables \texttt{born\_australia}: whether a person is born in Australia; \texttt{married}: marital status; \texttt{uni\_mem}: union membership; \texttt{gov\_emp}: government employment; \texttt{age}: age; \texttt{year\_expe}: work experience; \texttt{attitude}: index of labor market attitudes; \texttt{year\_edu}: education levels; \texttt{wage\_hour}:  hourly wage. Source indicates which questions in the survey provide the information.}
    \label{tab:ALS_sample1984}
\end{table}

\begin{table}[p]
    \centering
    \begin{tabular}{c c c c c c}
        \hline\hline
        Variable & Source & Mean & SD & Min & Max \\
        \hline
        \texttt{born\_australia} & B3                             & $0.84$ & $0.36$ & $0$ & $1$ \\
        \texttt{married}         & A7                             & $0.15$ & $0.36$ & $0$ & $1$ \\
        \texttt{uni\_mem}        & G11                            & $0.38$ & $0.49$ & $0$ & $1$ \\
        \texttt{gov\_emp}        & G10                            & $0.22$ & $0.42$ & $0$ & $1$ \\
        \texttt{age}             & A4                             & $20.22$ & $2.87$ & $15$ & $26$ \\
        \texttt{year\_expe}      & F3-4, F7-10, F31-33, F23-25    & $1.82$ & $2.13$ & $0$ & $16$ \\
        \texttt{attitude}        & O1-7                           & $18.75$ & $3.49$ & $6$ & $30$ \\
        \texttt{year\_edu}       & E3, E5, E8, E12, E14, E21, E23 & $11.69$ & $2.11$ & $2$ & $20$ \\
        \texttt{wage\_hour}      & G3-5, G7-8                     & $7.48$ & $2.94$ & $0.375$ & $75.00$ \\
        \hline\hline
    \end{tabular}
    \caption{The 1985 wave summary statistics of variables \texttt{born\_australia}: whether a person is born in Australia; \texttt{married}: marital status; \texttt{uni\_mem}: union membership; \texttt{gov\_emp}: government employment; \texttt{age}: age; \texttt{year\_expe}: work experience; \texttt{attitude}: index of labor market attitudes; \texttt{year\_edu}: education levels; \texttt{wage\_hour}:  hourly wage. Source indicates which questions in the survey provide the information.}
    \label{tab:ALS_sample1985}
\end{table}

\end{document}